\documentclass[prx,letterpaper,aps,superscriptaddress,floatfix,twocolumn,nofootinbib]{revtex4-2}
\pdfoutput=1

\usepackage{amsmath}
\usepackage{amssymb}
\usepackage{amsthm}
\usepackage{makecell}
\usepackage{amsfonts}
\usepackage{calc}
\usepackage[nomessages]{fp}
\usepackage{physics}
\usepackage{dsfont}
\usepackage{xcolor}
\usepackage{listings}
\usepackage{comment}
\usepackage{tabularx}
\usepackage{enumerate}
\usepackage{latexsym}
\usepackage{inputenc}
\usepackage{tikz}
\usetikzlibrary{patterns}
\usepackage{mathdots}
\usepackage{multirow}
\usepackage{thmtools}
\usepackage{thm-restate}
\usepackage{mathtools}

\usepackage{psfrag}

\usepackage{color}

\usepackage{bm}
\usepackage{graphicx}
\usepackage{subfigure}

\usepackage{graphicx}
\usepackage{amsmath}
\usepackage{amsfonts}
\usepackage{amssymb}
\usepackage{ulem}
\usepackage{xcolor}
\usepackage{slashed}
\usepackage{soul}
\usepackage{comment}
\usepackage{tikz-cd}
\usepackage[colorlinks=true, linkcolor=blue, citecolor=magenta]{hyperref}

\newcommand{\beq}{\begin{equation}}
\newcommand{\eeq}{\end{equation}}

\newcommand{\beqa}{\begin{eqnarray}}
\newcommand{\eeqa}{\end{eqnarray}}

\newcommand{\ea}{\end{array}}
 
\def\eea{\end{eqnarray}}

\def\<{\langle}
\def\>{\rangle}

\usepackage{amsmath}
\usepackage{comment}
\usepackage{amssymb}
\usepackage{amsthm}

\newtheorem{lemma}{Lemma}
\newtheorem{prop}{Proposition}

\newtheorem{cor}{Corollary}
\newtheorem{defin}{Definition}

\usepackage{color}
\usepackage{tikz-cd}
\usepackage{comment}

\usepackage{tikz}
\usepackage[export]{adjustbox}

\def\[#1\]{
  \begin{equation}\begin{gathered}#1\end{gathered}\end{equation}
}

\usepackage[safe]{tipa}

\newcommand{\eneq}{\end{equation}}

\def\ea{{\it et al.}}
\input{epsf}

\newcommand{\btp}{\begin{tikzpicture}}
\newcommand{\etp}{\end{tikzpicture}}

\newcolumntype{L}[1]{>{\raggedright\arraybackslash}p{#1}}
\newcolumntype{C}[1]{>{\centering\arraybackslash}p{#1}}
\newcolumntype{R}[1]{>{\raggedleft\arraybackslash}p{#1}}

\makeatletter
\newsavebox{\@brx}
\newcommand{\llangle}[1][]{\savebox{\@brx}{\(\m@th{#1\langle}\)}%
  \mathopen{\copy\@brx\kern-0.5\wd\@brx\usebox{\@brx}}}
\newcommand{\rrangle}[1][]{\savebox{\@brx}{\(\m@th{#1\rangle}\)}%
  \mathclose{\copy\@brx\kern-0.5\wd\@brx\usebox{\@brx}}}
\makeatother

\usepackage{wrapfig}

\usepackage{tcolorbox}
\tcbuselibrary{theorems}

\newtcbtheorem[number within=section]{mytheo}{Definition}%
{colback=green!5,colframe=green!35!black,fonttitle=\bfseries}{th}

\usepackage{float}
\usepackage{bbm}

\usepackage{booktabs}
\usepackage{makecell}

\makeatletter
\def\l@subsubsection#1#2{}
\makeatother

\begin{document}

\hfill MIT-CTP/6115

\title{
Disentangling the Toric Code
}

\author{Lei Gioia}
\affiliation{Department of Physics and Center for Theory of Quantum Matter, University of Colorado, Boulder, CO 80309, USA}
\affiliation{Walter Burke Institute for Theoretical Physics, Department of Physics, Caltech, Pasadena, CA, USA}
\author{Salvatore D. Pace}
\affiliation{School of Natural Sciences, Institute for Advanced Study, Princeton, NJ 08540, USA}
\affiliation{Department of Physics, Massachusetts Institute of Technology, Cambridge, MA 02139, USA}
\author{Ruben Verresen}
\affiliation{Pritzker School of Molecular Engineering, University of Chicago, Chicago, IL 60637, USA}
\author{Shu-Heng Shao}
\affiliation{Center for Theoretical Physics — a Leinweber Institute, Massachusetts Institute of Technology, Cambridge, MA 02139, USA}
\author{Ryan Thorngren}
\affiliation{Mani L. Bhaumik Institute for Theoretical Physics, Department of Physics and Astronomy,
University of California, Los Angeles, CA 90095, USA}

\begin{abstract}
Commuting-projector Hamiltonians with integer spectrum such as the 1+1d Ising ferromagnet and the 2+1d toric code may be regarded as the generators of non-on-site $U(1)$ symmetries. In this paper we study whether these symmetries can be made on-site in the presence of finite and infinite-dimensional ancillae. In the finite-dimensional case, we prove that the fractionalized excitations, such as domain walls and anyons, make it impossible to on-site these Hamiltonians. On the other hand, we show that when infinite dimensional ancillae (e.g., rotors) are allowed, both the 1+1d Ising ferromagnet and the 2+1d toric code can be disentangled, yielding on-site symmetries. This is not in contradiction with the non-trivial ground state order of these Hamiltonians, since the ancilla Hamiltonians are no longer bounded from below. Our results indicate agreement between lattice and quantum field theory anomalies in the presence of infinite-dimensional ancillae, while with finite-dimensional ancillae the obstruction to on-siteability may even be non-invertible.
\end{abstract}

\maketitle
\tableofcontents

\section{Introduction}
\label{sec:intro}

Whether a symmetry is anomalous or not has strong consequences for systems that host it, placing both constraints on their kinematic and dynamical behaviors. Although the theory of anomalous symmetries is well-understood in the context of quantum field theory, the concept of anomalous symmetries on the lattice is not yet understood at the same level. Developing the theory of anomalies on the lattice and its relationship to the continuum theory is thus a goal which promises to strengthen our understanding of the emergence of quantum field theory from the lattice~\cite{tHooft:1979rat,Harvey2005,McGreevy2023,Else:2014vma,Cheng:2022sgb}.

\begin{table*}[t]
    \centering
    \renewcommand{\arraystretch}{1.15}
    \setlength{\tabcolsep}{12pt}
    
    \begin{tabular*}{\textwidth}{
    @{\hspace{12pt}\extracolsep{\fill}}
    l c c
    @{\extracolsep{0pt}\hspace{12pt}}
}
    \toprule
    \textbf{Property}
        & $U(1)_{\mathrm{SSB}}$
        & $U(1)_{\mathrm{TC}}$ \\
    \midrule
    
    On-siteability with semi-bounded ancillae
        & \makecell[c]{
            \textbf{No} \\[-1pt]
            {\footnotesize Prop.~\ref{prop:ssbfiniteancilla}}
          }
        & \makecell[c]{
            \textbf{No} \\[-1pt]
            {\footnotesize Prop.~\ref{prop:tcrotors}}
          } \\
    
    \addlinespace[0.6em]
    
    On-siteability with unbounded ancillae
        & \makecell[c]{
            \textbf{Yes} \\[-1pt]
            {\footnotesize Prop.~\ref{prop:ssbonsiterotor}}
          }
        & \makecell[c]{
            \textbf{Yes} \\[-1pt]
            {\footnotesize
             Props.~\ref{prop:onsitingTC} and~\ref{prop:emonsite}}
          } \\
    
    \midrule
    
    SRE eigenstates
        & \makecell[c]{
            \textbf{Yes} \\[-1pt]
            {\footnotesize e.g., $\lvert 0\cdots 0\rangle$}
          }
        & \makecell[c]{
            \textbf{Yes} \\[-1pt]
            {\footnotesize
             Prop.~\ref{prop:TCSREfinitedim} and Cor.~\ref{cor:SRETC}}
          } \\
    
    \bottomrule
    \end{tabular*}    
    \caption{On-siteability (in finite depth) and short-range entangled (SRE) eigenstates for the $U(1)_{\rm SSB}$ symmetry, generated by the 
    1+1d Ising ferromagnetic Hamiltonian~\eqref{eq:HSSB}, and $U(1)_{\rm TC}$, generated by the 2+1d toric code Hamiltonian~\eqref{TCHam}. For both symmetries, on-siteability depends on the type of ancilla (semi-bounded versus unbounded, defined below in Def.~\ref{def:onsiteable}) allowed in the system. For $U(1)_{\rm SSB}$, the proven existence of SRE eigenstates is independent of the presence of ancillae. This differs for $U(1)_{\rm TC}$. In Proposition~\ref{prop:TCSREfinitedim} we construct a $U(1)_{\rm TC}$ symmetric SRE states in the presence of finite-dimensional fermionic ancillae. We also give a construction of an SRE eigenstate in Corollary~\ref{cor:SRETC} in the presence of rotors.}
    \label{tab:SSB}
\end{table*}

Anomalous symmetries are generally thought to obstruct the following three features: (1) there is a symmetric invertible state; (2) one can promote the global symmetry to a gauge symmetry with local commuting Gauss laws; and (3) there exists a finite-depth unitary transformation which makes the symmetry on-site. It is easy to show that ``on-siteability'' (3) implies both (2) and (1), and with usual methods of gauging (2) and (3) are essentially equivalent. In 1+1d, anomaly indices for finite groups have been defined \cite{Else:2014vma,ogata2021general,Kawagoe:2021gqi,Seifnashri:2023dpa,Kapustin:2024rrm}  and shown to vanish if and only if all three are satisfied, and thus all three sorts of lattice anomaly are equivalent \cite{seifnashri2025disentanglinganomalyfreesymmetriesquantum,bols2025classificationlocalitypreservingsymmetries}, representing a complete solution in this dimension which is in concert with quantum field theory.

On the other hand, recent papers \cite{Tu_2026,shirley2025anomalyfreesymmetriesobstructionsgauging,Kapustin:2025nju,Kawagoe:2025ldx,czajka2025anomalieslatticehomotopyquantum} have shown that in 2+1d this story is more complicated, and there exist symmetries that only satisfy a proper subset of the above qualities. For example, certain $\mathbb{Z}_2$ symmetries have been shown to possess symmetric invertible states (1) but cannot be made on-site using finite-dimensional ancillae (3)~\cite{shirley2025anomalyfreesymmetriesobstructionsgauging}. The obstruction to (3) in these examples is based on the fact that a lattice translation cannot be expressed as a finite-depth quantum circuit. However, in a continuum limit, this lattice translation disappears, which is why this lattice obstruction is invisible in quantum field theory.

In fact, there is also a way to trivialize these obstructions on the lattice, by using infinite-dimensional ancillae, for which lattice translations become finite-depth circuits, as shown in Ref.~\cite{Jones:2026gmd}. 
In the same work, they also show that with infinite-dimensional ancillae, one can derive direct correspondences between invertible states on the lattice and certain quantum field theory algebras. One may likewise expect that the theory of lattice anomalies may more closely reflect quantum field theory once we allow such ancillae.

In this paper, we will test this premise by studying properties (1), (2), and (3) for a class of $U(1)$ symmetries $e^{i\theta H}$ where $H$ is a commuting-projector Hamiltonian (these were called pivot Hamiltonians in \cite{SciPostPhys.14.2.012}). Our main examples are the $U(1)$ symmetries generated by the 1+1d Ising ferromagnet Hamiltonian $H_{\rm SSB}$ and by the 2+1d toric code Hamiltonian $H_{\rm TC}$. We show that since the ground state subspaces of these Hamiltonians are in a particular sense non-trivial (they host non-factorizable operators), these Hamiltonians are not on-siteable (3) with finite-dimensional ancillae. On the other hand, we will show that these Hamiltonians admit short-range entangled eigenstates, meaning that with respect to $e^{i \theta H}$ there exist symmetric trivial states (1). Unlike the 2+1d examples mentioned above, the obstruction to (3) is not related to lattice translations. Stranger still, with finite-dimensional ancillae there is no other $U(1)$ symmetry $e^{i \theta H'}$ such that $e^{i \theta H} \otimes e^{i \theta H'}$ is on-siteable---in this sense the anomaly is ``non-invertible''. Thus, there is no group theory based classification (e.g., cohomology theory) for these obstructions to on-siteability with only finite-dimensional ancillae.

We show that this strange behavior in the case of finite-dimensional ancillae is mollified in the presence of infinite-dimensional ancillae. In particular, both $H_{\rm SSB}$ and $H_{\rm TC}$ become on-siteable. This is only possible since the obstruction defined above is based on the ground state subspace, which is no longer well-defined if we tensor $H$ with an unbounded on-site symmetry generator like the rotation action on a rotor. On-siteability of the toric code is particularly non-trivial, and agrees with the quantum field theory prediction that there are no anomalous $U(1)$ symmetries in 2+1d (equivalently no non-trivial $U(1)$ SPT phases in 3+1d \cite{kapustin2014bosonictopologicalinsulatorsparamagnets}). The summary of results is given in Table \ref{tab:SSB}.

The paper is organized as follows: in Section~\ref{sec:onsitegenmet} we will define the relevant concepts, provide intuition into the problem, and prove general properties regarding the on-siteability of the $U(1)$ symmetries, generated by commuting-projector Hamiltonians, in the presence of both \textit{semi-bounded} (a generalization of finite-dimensional) or infinite-dimensional \textit{unbounded} ancillae. We will then apply our proofs to study both the obstruction to on-siteability for semi-bounded ancillae, as well as explicit construction of disentangling circuits with unbounded ancillae for both the SSB symmetry and a family of toric code symmetries in Sections~\ref{sec:SSB} and \ref{sec:toriccode}, respectively. Finally, in Section~\ref{sec:discussion}, we conclude with a discussion and open questions on this subject.

\section{On-siteability, general discussion}
\label{sec:onsitegenmet}

In this section we present some intuition and general results that underlie the physics of when a $U(1)$ symmetry generated by a commuting-projector Hamiltonian can or cannot be disentangled.

On-siteability may be defined for any symmetry group $G$ acting as unitaries on a tensor product Hilbert space:
\begin{defin} {\rm (On-siteable in depth $D$)}
\label{def:onsiteable}
    A $G$-symmetry which acts by unitary operators $\{U(g)\}_{g\in G}$ is `on-siteable in depth $D$'\footnote{For a finite size system, we require $D$ to be independent of the system size such that $D$ remains finite in the thermodynamic limit.} if there exists a quantum circuit $W=\prod_{j=1}^D W_j$ of depth $D$, with $W_j$ being a product of finite-range non-overlapping unitaries\footnote{Here the disentangling unitary is a quantum circuit, as opposed to the more general quantum cellular automaton (QCA)~\cite{Gross_2012}. However, these two objects are actually equivalent in the presence of ancillae (of the same on-site Hilbert space dimension as the original model), as the QCA $\alpha$ can be made into a quantum circuit $\alpha\otimes\alpha^{-1}$ with $\alpha^{-1}$ acting on ancillae.
    Then, if a nontrivial QCA $\alpha$ makes every $U(g)$ on-site, the circuit  $\alpha\otimes\alpha^{-1}$ also makes every $U(g)\otimes 1$ on-site.}, such that $\forall g\in G$,
    \begin{align}\label{eqnonsitable}
        W \left(U(g)\otimes V(g)\right) W^\dag=V(g)'\quad,
    \end{align}
where $V(g)$ and $V(g)'$ are on-site symmetries, acting as tensor products of unitary $G$ representations. The ancillae on which $V(g)$ acts may be constrained to be finite-dimensional or allowed to be infinite-dimensional, in which case we will specify the symmetry is on-siteable with (in)finite-dimensional ancillae. Below, for $G = U(1)$ we also discuss semi-bounded vs. unbounded ancillae.
\end{defin}
A symmetry that is on-siteable in finite depth is also gaugeable and has an SRE eigenstate\footnote{In this paper, a state $\ket{\psi}$ is an SRE state if there exists a FDQC $U$ and product state ${\bigotimes_j \ket{\phi_j}}$ such that ${\ket{\psi} = U \bigotimes_j \ket{\phi_j}}$. SRE states are a subset of invertible states.}, thereby implying the absence of any of the three types of lattice anomalies described in Section~\ref{sec:intro}. 
Indeed, the Gauss operators needed to gauge an on-site symmetry ${\{V(g)' = \prod_i V_i(g)'\}_{g\in G}}$ can be constructed directly from its local operators $V_i(g)'$~\cite{shirley2025anomalyfreesymmetriesobstructionsgauging}. Furthermore, starting from a symmetric product state ${\bigotimes_i \ket{\mathbf{1}}_i}$ of the on-site symmetry ${\{V(g)' = \prod_i V_i(g)'\}_{g\in G}}$, there always exists the SRE state ${W^\dag\bigotimes_i \ket{\mathbf{1}}_i}$ which is a symmetric SRE state of ${\{U(g)\otimes V(g)\}_{g\in G}}$.

In this paper, we will consider on-siteability for $U(1)$ symmetries, in the presence of two types of ancilla (1) \textit{semi-bounded ancillae}: ancillae with symmetry generators (or charges) that are all bounded from below and/or above, or (2) \textit{unbounded ancillae}: ancillae with symmetry generators that are unbounded from \textit{both} above and below. Semi-bounded ancilla automatically includes finite-dimensional ancilla, and unbounded ancilla are necessarily infinite-dimensional. Finally, in the case that a symmetry is not on-siteable, let us define this obstruction to be \textit{invertible} if a weaker notion of Definition~\ref{def:onsiteable} holds, where $V(g)$ is allowed to be a finite-depth quantum circuit but still a unitary $G$ representation.\footnote{Note that this notion is also ancillae type dependent. The definition is motivated by the idea that the invertibility of non-on-siteability is expected to be related to the invertibility of anomalies.}

We explore the various notions of lattice anomalies described above for $U(1)$ symmetries. Such symmetries have a Hermitian generator and can be expressed as $U(\theta)= e^{i\theta H}$. $\{U(\theta)\}_{\theta \in [0,2\pi)}$ is on-siteable in depth $D$ if and only if this generator (or charge) $H$ can be disentangled to an on-site generator in depth $D$. In particular, this means that
$$W ( H \otimes 1 + 1 \otimes H_0 ) W^\dagger = H'_0\,,$$
where $W$ is the same as in Eq.~\ref{eqnonsitable}, and $H_0$ and $H_0'$ are on-site generators (sum of support-one Hermitian operators with integer-quantized spectrum), acting on the ancillae and the entire system, respectively.
This connection between the symmetry and generator is formally proven in Appendix~\ref{app:cononsite}.

This allows us to study the disentanglability of a special class of $U(1)$ generators which correspond to Hamiltonians with integer-quantized spectra. We will specifically be interested in commuting-projector Hamiltonians, for which it is easy to compute the spectrum.

Some of these Hamiltonians can be straightforwardly disentangled in finite-depth, such as the cluster state parent Hamiltonian $$H=-\sum_j Z_{j-1} X_j Z_{j+1},$$ which has the disentangling finite-depth quantum circuit (FDQC) $\prod_j CZ_{j,j+1}:Z_j\mapsto Z_j,X_j\mapsto Z_{j-1} X_j Z_{j+1}$ \cite{Briegel01}. Another example is the Kitaev chain Hamiltonian \cite{Kitaev_2001} $$H=i\sum_j \gamma_{jB}\gamma_{j+1A},$$ which can be disentangled via a QCA Majorana-translation $T_{\frac{1}{2}}\colon \gamma_{jB}\mapsto \gamma_{j+1 B}, \gamma_{jA}\mapsto\gamma_{jA}$. It can also be disentangled by a circuit if we add a layer of ancillae on which $H$ acts by 0 and apply an inverse Majorana translation to the new layer. The combined counter translations are a finite depth circuit \cite{Gross_2012}.

The above examples are Hamiltonians with unique, invertible ground states. We will introduce new obstructions to on-siteability for commuting-projector Hamiltonians with degenerate ground state spaces and fractionalized excitations, which apply for semi-bounded ancillae, for which ground states\footnote{Or one can consider ``sky states'' of maximum energy.}  are well-defined.

The basic idea can be understood by examining the example of the spontaneous symmetry-broken (SSB) 1+1d quantum Ising Hamiltonian on a lattice $\Lambda = \{0,\ldots,N-1\}$ with periodic boundary conditions $N \sim 0$:
\begin{align}
    H_{\rm SSB}=\sum_{j\in\Lambda} \frac{1}{4}\left(1-Z_j Z_{j+1}\right)\quad.
    \label{eq:HSSB}
\end{align}
This Hamiltonian has been normalized so that in the $Z$ basis, its value is the number of domain walls divided by two. Note that there is always an even number of domain walls on a periodic chain, so $H_{\rm SSB}$ has integer spectrum. We choose this normalization to reflect that a single domain wall is a fractional excitation, which can only be created in pairs by local operators.

Related to this, we may consider the ground state subspace of $H_{\rm SSB}$: $\mathcal{S}_{\rm GS}={\rm Span}\{\ket{0\cdots 0},\ket{1 \cdots 1}\}$ with long-range entangled (LRE) ground states $\ket{\rm GHZ_\pm}=\frac{1}{\sqrt{2}}(\ket{0 \cdots 0}\pm\ket{1 \cdots 1})$\footnote{The LRE nature follows since one may transform this state via an FDQC to a state with a non-zero momentum, which are all LRE~\cite{PhysRevX.12.031007}.}. One may suspect that such a Hamiltonian cannot be on-sited due to the existence of LRE ground states. However, this logic is not straightforward since even trivial on-site Hamiltonians, such as $H=0$, have LRE ground states. One could hypothetically argue that the $H_{\rm SSB}$ LRE ground states could be disentangled via FDQC to LRE ground states that arise from on-site $h=0$ terms of the on-sited Hamiltonian. Of course, one may then ask, what about the excited states? This is precisely why we need to consider the fractionalized excitations of such models.

To consider these excitations, observe that the operator $\prod_{j\in\Lambda} X_j$ preserves $\mathcal{S}_{\rm GS}$, but cannot be factorized into $\mathcal{S}_{\rm GS}$-preserving operators along any bipartition of $\Lambda$ into contiguous regions $A$ and $B$, i.e., cannot be written as $\prod_{j\in\Lambda} X_j=\sum O_A O_B$ such that $O_A$ and $O_B$ are not proportional to the identity and each preserve $\mathcal{S}_{\rm GS}$. For example, the naive factorization $O_{A/B}=\prod_{j\in A/B} X_j$ creates domain-wall excitations at the ends of $A/B$. Under a disentangling FDQC of $H_\mathrm{SSB}$, $\prod_{j\in\Lambda} X_j$ would be related via the FDQC to another ground state subspace preserving operator of the on-site Hamiltonian. However, all such operators in an on-site Hamiltonian are factorizable, resulting in a direct contradiction if regions $|A|,|B|>\alpha D$ for some constant $\alpha$ and circuit depth $D$. The same logic applies to other systems with fractionalized excitations, and is formalized in Proposition \ref{propfactorization} below.

Crucially, this intuition relies on the presence of a well-defined ground state subspace (or, alternatively, highest-excited sky state subspace), and its associated algebra of ground state operators (defined precisely below).
Therefore, infinite-dimensional unbounded ancilla (such as rotors) can circumvent this obstruction, as we will show in Prop.~\ref{prop:ssbonsiterotor}. The main results for the $U(1)_{\rm SSB}$ symmetry implemented by $U_{\rm SSB}(\theta)=e^{i \theta H_{\rm SSB}}$ are shown in Table~\ref{tab:SSB}. Our results are consistent with what might be expected for continuous symmetries in 1+1d, as they indicate a generalization of the formulation for finite symmetry $G$ which required ancillae living in $L^2(G)$ to be on-siteable~\cite{seifnashri2025disentanglinganomalyfreesymmetriesquantum}. Thus, we expect that when $G$ is continuous, on-siteability generally requires infinite-dimensional ancillae containing a copy of each irrep of $G$, which, in the examples studied here, we show must also be unbounded.

Can we similarly disentangle Hamiltonians in topological phases in higher dimensions, such as toric code in 2+1d~\cite{KITAEV20032}? Due to the presence of non-trivial braiding statistics, the case for topological orders is more complicated compared to the SSB case even when unbounded ancillae are taken into account. Perhaps surprisingly, despite their  differences, we will show that the toric code Hamiltonian and its extended family are on-siteable in the presence of unbounded ancilla. That is, the $U(1)$ global symmetry $e^{i \theta H_{\rm TC}}$ in 2+1d generated by the toric code Hamiltonian $H_{\rm TC}$ can be made on-site with these ancillae. This is consistent with the expectation from continuum quantum field theory, where there is no anomaly for a $U(1)$ global symmetry in 2+1d. We stress that these disentangled Hamiltonians cannot be regarded as the Hamiltonian of a physical system, since their energy spectra are unbounded; nonetheless, they serve as generators of compact, locality-preserving $U(1)$ global symmetries.

\subsection{Semi-bounded ancilla}
\label{sec:semibounded}

To formalize the obstruction to on-siteability intuition presented in Sec.~\ref{sec:intro} in the presence of semi-bounded ancillae\footnote{In the following, we may assume that the ancilla generators are bounded from below. If the generator $H$ was instead bounded from above, we can replace it with $-H$ which is bounded from below.}, we will first introduce the notion of a \textit{ground state algebra} (adapted from Def. 2.9 of ~\cite{Jones_Naaijkens_Penneys_Wallick_Izumi_2025}).

\begin{defin}{\rm(Ground state algebras)}
    Let $H$ be a gapped\footnote{We allow degeneracy. ``Gapped'' means that the lowest eigenvalue is isolated in the spectrum of $H$, which ensures the existence of the ground state projector $P$.} Hamiltonian and $P$ its ground state projector. We define the ``ground state algebra'' $\mathcal{A}_{\rm gs}$ of $H$ as the algebra of operators $xP$ where $x$ is an operator satisfying $xP = Px$. We can think of $x$ as an operator preserving the ground state space, since $xP = Px$, while $xP = yP$ if $x$ and $y$ have the same action on the ground states. Note that the identity of $\mathcal{A}_{\rm gs}$ is $P$. We also define, for each $X$ the subalgebra of $\mathcal{A}_{\rm gs}$ of ground state operators localizable to $X$
    \[\mathcal{A}_{\rm gs}(X) = \{xP \ |\ x\text{ supported in }X,xP=Px\}.\]
    These satisfy
    \begin{enumerate}
    \item if $X \subset Y$, $\mathcal{A}_{\rm gs}(X) \subset \mathcal{A}_{\rm gs}(Y)$
    \item if $X \cap Y = \varnothing$, $a \in \mathcal{A}_{\rm gs}(X)$, $b \in \mathcal{A}_{\rm gs}(Y)$ commute:
    \[ab=ba.\]
\end{enumerate}
Thus we may regard $X \mapsto \mathcal{A}_{\rm gs}(X)$ as a local net of algebras~\cite{Jones_Naaijkens_Penneys_Wallick_Izumi_2025,Harlow:2025cqc,Holfester:2026peg}.
\end{defin}
For example, consider the ferromagnetic Hamiltonian~\eqref{eq:HSSB}. The projector
\[P = \prod_{j \in \Lambda} \frac12 (1+Z_j Z_{j+1})\]
yields a ground state algebra generated by
\[ \bigg( \prod_{j \in \Lambda} X_j \bigg) P \qquad Z_1 P\]
which is globally isomorphic to a $2\times 2$ matrix algebra but has an interesting locality structure. For example, $Z_1 P$ is localizable to any spin because $Z_1 P = Z_j P$ for all $j \in \Lambda$. Meanwhile, $\prod_{j \in \Lambda} X_j$ is only localizable to the whole chain $X=\Lambda$.

A bounded-spread isomorphism $\varphi$ between two local nets $\varphi:\mathcal{A} \to \mathcal{B}$ is an isomorphism of algebras $\mathcal{A}(\Lambda) \to \mathcal{B}(\Lambda)$ with a spread $s$ such that
\[\varphi(\mathcal{A}(X)) \subset \mathcal{B}(X^{+s}) \\
\varphi^{-1}(\mathcal{B}(X)) \subset \mathcal{A}(X^{+s})\]
where $X^{+s}$ denotes all the points within a distance $s$ of $X$.

\begin{lemma}\label{propbspreadiso}
    Suppose $H$ is a gapped Hamiltonian and $\alpha$ is a spread-$s$ QCA. Then $\alpha$ induces a spread-$s$ isomorphism between the ground state algebra $\mathcal{A}_1$ of $H$ and the corresponding $\mathcal{A}_2$ of $\alpha(H)$.
\end{lemma}

\begin{proof}
Since $\alpha$ is a map of algebras, if
\[xP \in \mathcal{A}_1(X)\]
then
\[\alpha(xP) = \alpha(x) \alpha(P) \in \mathcal{A}_2(X^{+s}).\]
The inverse map is given by $\alpha^{-1}$, which also has spread $s$ \cite{Gross_2012,Haah_2023}.
\end{proof}

\begin{lemma}
    Suppose $H_1$ and $H_2$ are gapped Hamiltonians with ground state algebras $\mathcal{A}_{1}$ and $\mathcal{A}_{2}$. The ground state algebra of $H_{12}=H_1 \otimes 1 + 1 \otimes H_2$ satisfies
    \[\mathcal{A}_{12}(X)=\mathcal{A}_{1}(X) \otimes \mathcal{A}_{2}(X)\]
\end{lemma}
\begin{proof}
    If $P_1, P_2$ are the ground state projectors of $H_1,H_2$, then that of $H_{12}$ is $P_1 \otimes P_2$.
\end{proof}

With these two lemmas, we will find that the ground state algebra of an on-siteable Hamiltonian thus behaves quite like a tensor product of local algebras. To exploit this, we define a factorization of an operator.
\begin{defin}{\rm(Factorizable operator)}
    We say that an operator $O \in \mathcal{A}_{\rm gs}(X)$ factorizes over a bipartition $X = Y \cup Z$ in spread-$r$ if there are operators $a_i \in \mathcal{A}_{\rm gs}(Y^{+r})$ and $b_i \in \mathcal{A}_{\rm gs}(Z^{+r})$ such that
\[O = \sum_i a_i b_i\,,\]
(recall here equality means equivalent action on the ground state space).
\end{defin}
Note that in this definition, the support of $a_i$ and $b_i$ are allowed to overlap, but only in a region $Y^{+r} \cap Z^{+r}$.

For a tensor product algebra, every operator admits a factorization over any bipartition $X = Y \cup Z$ in spread 0. In particular, this will be true for the ground state algebra of any on-site semibounded Hamiltonian. For Hamiltonians which are on-siteable by spread $s$ QCA, this will be true in spread-$2s$, leading to the following:

\begin{prop}\label{propfactorization}{\rm(Factorization in the ground state algebra of on-siteable Hamiltonians)}
    Suppose $H$ is gapped and on-siteable via a spread-$s$ QCA. Then all operators $a \in \mathcal{A}_{\rm gs}(X) \otimes \mathcal{A}_{\rm gs}^0(X)$ admit factorizations over any bipartition $X = Y \cup Z$ in spread-$2s$, where $\mathcal{A}_{\rm gs}^{0}$ is the ground state algebra of the ancilla Hamiltonian $H_0$. In particular, if there are operators in $\mathcal{A}_{\rm gs}$ which are not arbitrarily factorizable in spread-$2s$, then $H$ is not on-siteable via a spread-$s$ QCA.
\end{prop}
\begin{proof}
Let $\mathcal{A}_{\rm gs}$, $\mathcal{A}_{\rm gs}^0$, $\mathcal{A}_{\rm gs}^{0'}$ be the ground state algebras of $H$, $H_0$, $H_0'$, with
\[H_0' = \alpha(H \otimes 1 + 1 \otimes H_0)\]
on-site, with $\alpha$ a spread-$s$ QCA, and all three gapped (and bounded from below). By Lemmas 1 and 2, we obtain a spread-$s$ isomorphism
\[\alpha:\mathcal{A}_{\rm gs}(X) \otimes \mathcal{A}_{\rm gs}^0(X) \to \mathcal{A}_{\rm gs}^{0'}(X^{+s}).\]
Since $H_0'$ is composed of single-site Hamiltonians, $\mathcal{A}_{\rm gs}^{0'}(X^{+s}) = \bigotimes_{x \in X^{+s}} \mathcal{A}_{\rm gs}^{0'}(\{x\})$. In particular, for any $a\in \mathcal{A}_{\rm gs}(X) \otimes \mathcal{A}_{\rm gs}^0(X)$ and any bipartition $X = Y \cup Z$ we can express
\[\alpha(a) = \sum_i b_i c_i\]
with $b_i \in \mathcal{A}_{\rm gs}^{0'}(Y^{+s})$ and $c_i \in \mathcal{A}_{\rm gs}^{0'}(Z^{+s})$. Inverting, we get
\[a = \sum_i \alpha^{-1}(b_i) \alpha^{-1}(c_i)\]
with $\alpha^{-1}(b_i) \in \mathcal{A}_{\rm gs}(Y^{+2s}) \otimes \mathcal{A}_{\rm gs}^0(Y^{+2s})$ and $\alpha^{-1}(c_i) \in \mathcal{A}_{\rm gs}(Z^{+2s}) \otimes \mathcal{A}_{\rm gs}^0(Z^{+2s})$. Thus, all $a\in \mathcal{A}_{\rm gs}(X)$ admit arbitrary spread-$2s$ factorizations. The statement follows.
\end{proof}

Note that the proposition holds even when we include semi-bounded ancillae that are not necessarily integer-charged under the Hamiltonian generated symmetry.

Proposition~\ref{propfactorization} formalizes the intuition for the spontaneous symmetry-broken 1+1d quantum Ising model example in Sec.~\ref{sec:intro} since symmetry operator $\prod_{j\in\Lambda} X_j$ is non-factorizable, as we will show formally in Sec.~\ref{sec:SSB}. In subsequent sections, we will also see that other Hamiltonians with fractionalized excitations, such as topological orders, also possess ground state algebra operators that do not factorize, thereby leading to an obstruction in disentangling the Hamiltonian in the presence of ancilla with semi-bounded generators $H_0$. Again, notice that the proof of Proposition~\ref{propfactorization} relies crucially on the ancilla possessing generators (Hamiltonians) that have a ground state space, because it uses the structure of the ground state algebra. This observation naturally leads to questions regarding infinite-dimensional ancilla with associated unbounded generators.

\subsection{Unbounded ancilla}

In fact, we can provide a sufficient and necessary condition for disentangling a Hamiltonian in the presence of ancilla living in $L^2(G)$, i.e., rotors, with unbounded generators when $G$ is infinite. To do this, let us first introduce the concept of a locally factorizable symmetry.
\begin{defin}{\rm(Local factorizability)}
    The $G$ symmetry operators $\{U(g)\}_{g\in G}$ are
    locally factorizable if every $U(g)$ can be written as a product of mutually commuting unitary $G$-representations $U_i(g)$ of bounded size, such that $U(g) = \prod_i U_i(g)$ and $[U_i(g),U_j(h)]=0$ $\forall g,h\in G,\forall i\neq j$.
\end{defin} 

In particular, a $U(1)$ symmetry is locally factorizable if and only if its generator $H$ can be written as a sum of local terms
\begin{align}
    H=\sum_{S\subset\Lambda} (h_S+h_{\partial S})\quad,
    \label{eq:hboundaryfactorizability}
\end{align}
such that the term in the parentheses is supported on a local subset $S\subset\Lambda$ and is a local $U(1)$ representation $e^{i2\pi( h_S+h_{\partial S})}=1$, and the boundary term $h_{\partial S}$ is supported on $\partial S$ (the boundary of $S$) with $\sum_{S\in\Lambda} h_{\partial S}=0$. Commutativity also implies that the terms in the parentheses must also mutually commute with each other. If the generator fulfills the above condition, we will also refer to it as locally factorizable.

In fact, local factorizability of $G$ symmetry operators in the presence of $L^2(G)$ ancillae is equivalent to their on-siteability in the presence of $L^2(G)$ ancillae.
\begin{prop} {\rm(On-siting via unbounded ancilla)}
    $U(g)$ is on-siteable with the addition of ancilla living in $L^2(G)$ carrying the on-site left-regular $G$ action if and only if $\{U(g)\}_{g\in G}$ is locally factorizable in the presence of ancilla living in $L^2(G)$ carrying the on-site left-regular $G$ action. 
    \label{prop:disentanglerlemma}
\end{prop}
\begin{proof}
    First, assume that $\{U(g)\}_{g\in G}$ is locally factorizable in the presence of ancilla living in $L^2(G)$ carrying the on-site left-regular $G$ action. Then, every $U(g) = \prod_i U_i(g)$. For each $i$, introduce an ancilla $\ket{g_i} \in L^2(G)$. Define the operator
    \[V_i = \int dg_i \, U_i(g_i)
    \otimes
    |g_i \rangle \langle g_i|\]
    It is easy to check $V_i$ is unitary and $V_i$ and $V_j$ commute when $i \neq j$. Furthermore, letting $L(h)\ket{g_i} = \ket{hg_i}$ be the on-site left regular action, we have
    \[V_i L(h) V_i^\dagger \ket{g_i} = U_i(h g_i) U_i(g_i^{-1}) \ket{h g_i} = U_i(h) \ket{hg_i}\,.\]
    Thus,
    \[(\prod_i V_i) \left[1\otimes L(h)\right] (\prod_i V_i)^\dagger = U(h) \otimes L(h)\,,\]
    so $\prod_i V_i^\dag$ (which can be expressed as a finite-depth circuit assuming $U_i$ has bounded size) makes $U(h) \otimes L(h)$ on-site.

    Conversely, let us assume that $U(g)$ is on-siteable with the addition of ancilla living in $L^2(G)$ carrying the on-site left-regular $G$ action. Following the notation in Eq.~\ref{eqnonsitable} and expanding the on-sited symmetry $V(g)'=\prod_i V_{i}(g)'$,
    $$U(g)\otimes V
    (g)=\prod_i W^\dag V_{i}(g)'W=\prod_i U_i(g)\,,$$
    with $U_i(g)\equiv W^\dag V_i(g)'W$. So we have found a local factorization in the presence of ancilla living in $L^2(G)$ carrying the on-site left-regular $G$ action.
\end{proof}

In the following Secs.~\ref{sec:SSB} and \ref{sec:toriccode}, we will show how this explicit construction can be used to disentangle the Ising ferromagnetic Hamiltonian and toric code Hamiltonian, and thereby on-site their associated symmetries in the presence of rotors. Notice that it is crucial for $U(g)$ to be a product of local mutually-commuting $G$-representations, which may not always be immediately clear from the outset. In Appendix~\ref{sec:blendability} we study how this requirement is intimately linked to blendability of the symmetry~\cite{czajka2025anomalieslatticehomotopyquantum,Tu_2026}.

\section{Ising ferromagnetic generator}
\label{sec:SSB}

Recall from Eq.~\ref{eq:HSSB}, $U(1)_{\rm SSB}$ is enacted by the unitary operator
\[U_{\rm SSB}(\theta) = e^{i\theta H_{\rm SSB}}\quad,\]
which has been previously studied in Refs. \cite{Vernier:2018han,PhysRevLett.134.021601,SciPostPhys.14.2.012,Pace:2024oys,Jones_2025}. 
Notice that $H_{\rm SSB}$ is normalized in such a fashion that the charge generator is quantized in $\mathbb{Z}$, and $U_{\rm SSB}(2\pi)=1$. This symmetry is a finite-depth quantum circuit which,  via conjugation, maps local operators to other local operators. Explicitly, the operators $X_j$ and $Z_j$ transform under $U(1)_{\rm SSB}$ as
\begin{align}
    Z_j\mapsto Z_j\,, && X_j\mapsto X_j e^{i\frac{\theta}{2}(Z_{j-1}Z_j+Z_j Z_{j+1})}\,,
\end{align}
where the spread of the operator support is only to its nearest neighbors.

First, note that even in the absence of ancilla, $U(1)_{\rm SSB}$ possesses short-range entangled eigenstates such as $\ket{0...0}$. This feature means that it is not anomalous from an `obstruction to an SRE state' perspective. 
However, we will encounter an obstruction to making this symmetry on-site, which is another expected feature of an anomaly-free symmetry, using a large class of ancillae.

To understand this, recall the kinematic properties of the SSB Hamiltonian mentioned in the beginning of Section~\ref{sec:onsitegenmet}. The SSB Hamiltonian has ground state subspace $\mathrm{Span}\{\ket{0...0}, \ket{1...1}\}$ with fractionalized excitations, known as domain walls, which necessarily arise in pairs in a system with periodic boundary conditions.

Applying the general methods in Sec.~\ref{sec:onsitegenmet}, one can show that $U(1)_{\rm SSB}$ is not on-siteable via an FDQC due to the $\prod_{j\in\Lambda} X_j$ operator being non-factorizable.
\begin{prop}
    $U(1)_{\rm SSB}$ can be on-sited only via a 
    QCA of spread $\geq L/8$, where $L$ is the system size, even in the presence of semi-bounded ancilla\footnote{Spread $s \ge L/8$ requires a disentangling circuit to have depth at least $aL+b$, with $a,b$ universal constants depending on the size of the allowed circuit elements.}. 
    \label{prop:ssbfiniteancilla}
\end{prop}
\begin{proof}
    Consider the ground state algebra of $H_{\rm SSB}\otimes 1+1\otimes H_0$. It is a tensor product of the ground state algebras for $H_0$ as well as that for $H_{\rm SSB}$. Recall from Section~\ref{sec:semibounded} that the ground state algebra of $H_{\rm SSB}$ is a $2 \times 2$ matrix algebra, generated by $\prod_{j\in\Lambda} X_j P$ and $Z_j P$ where $j$ is any particular site and $P$ is the $H_{\rm SSB}$ ground state projector.

    Suppose that $\prod_{j \in \Lambda} X_j \otimes 1$ admitted a factorization
    \[\prod_{j \in \Lambda} X_j = \sum_i a_i b_i\]
    for $a_i$ localizable to $[1,L/2]^{+2s}$ and $b_i$ localizable to $[L/2+1,L]^{+2s}$. If $s < L/8$, then these sets individually do not contain all points of $\Lambda = [1,L]$. Therefore, each $a_i$ and $b_i$ commutes with $Z_j P \otimes 1$ since $Z_jP \otimes 1$ is localizable to a disjoint site for each $a_i$ or $b_i$. This is impossible because $Z_j P \otimes 1$ must anti-commute with $\prod_{j \in \Lambda} X_j$. Thus, $s \ge L/8$. By Proposition~\ref{propfactorization}, this implies that  $U(1)_{\rm SSB}$ is not on-siteable by a 
    QCA of spread ${<L/8}$.
\end{proof}
This immediately implies that $U(1)_{\rm SSB}$ is not on-siteable with FDQC nor any locality-preserving unitary. Note that the proof of this theorem does not rely on the ancilla Hamiltonian $H_0$ being on-site or integer-charged, only gapped. This means that the non-on-siteability of $U(1)_{\rm SSB}$ is also \textit{non-invertible}.

In fact, for certain system sizes and \textit{finite-dimensional} ancilla, we can show that $U(1)_{\rm SSB}$ cannot be disentangled via \textit{any} unitary, regardless of circuit depth.
\begin{prop}
    On system sizes $L=2^n$ for $n\in\mathbb{N}$ and $n\gg1$, $U(1)_{\rm SSB}$ cannot be on-sited via any unitary, even in the presence of sites\footnote{A site can include both the the original model site and that of an ancilla.} of dimension $<2^{L-1}$.
    \label{prop:galoisssb}
\end{prop}
The proof of this proposition is in Appendix~\ref{app:galois}, and shows that the entire energy spectrum, containing fractionalized excitations, cannot be reproduced by any on-site Hamiltonian with finite-dimensional local Hilbert spaces. The argument encodes the degeneracy of the spectrum in terms of polynomials and reduces the on-siteability feature to a problem of determining irreducible polynomials.

The proofs of Propositions \ref{prop:ssbfiniteancilla} and \ref{prop:galoisssb} rely explicitly on semi-bounded ancilla. In the presence of unbounded ancilla, $U(1)_{\rm SSB}$ is actually on-siteable.
\begin{prop}
    $U(1)_{\rm SSB}$ is on-siteable in finite depth in the presence of infinite-dimensional ancilla with unbounded generators.
    \label{prop:ssbonsiterotor}
\end{prop}
\begin{proof}
With periodic boundary conditions, we first note that $U(1)_{\rm SSB}$ is locally factorizable as it can be expressed as a product of the locally commuting $U(1)$ representations
\[
U_i(\theta) = e^{i\theta \frac12 (1+Z_i) \frac12 (1-Z_{i+1})}\quad.
\]
(Note that in the product $\prod_i U_i(\theta)$, the extra terms ${Z_i - Z_{i+1}}$ telescope.) 
Importantly, each $U_i(\theta)$ generates a $U(1)$ symmetry as $U_i(2\pi)=1$. 
Using this new expression, we can then apply Proposition~\ref{prop:disentanglerlemma}, which implies that $U(1)_{\rm SSB}$ is on-siteable.

More explicitly, the disentangler can be constructed as follows. 
We add a rotor $\phi_i\in \mathbb{R}/\mathbb{Z}$ at every site $i$. Now we add to the SSB Hamiltonian the ancillary rotors
\[
H = H_{\rm SSB} \otimes 1 + \sum_i 1\otimes N_i
\]
where $N_i= - {i\over 2\pi} {\partial \over \partial \phi_i}$ is the integer-valued operator conjugate to $\phi_i$. 
Define the FDQC 
\[
W = \prod_i e^{2\pi i \phi_i P_i},~~~
P_i  = \frac{1+Z_i}{2}  \frac{1-Z_{i+1}}{2},
\]
which is well-defined since $P_i$ takes integer values. 
Notice that
\[
W HW^\dag = \sum_i 1\otimes N_i
\]
which means that $U(1)_{\rm SSB}$ is made on-site by $W$.
\end{proof}

In the following section we will see that satisfying these two clauses is non-trivial for topological orders such as the toric code.

\section{Toric code generator}
\label{sec:toriccode}

We define the non-on-site toric code  $U(1)_{\rm TC}$ symmetry in 2+1d, generated by the toric code Hamiltonian \cite{KITAEV20032}
\begin{equation}\label{TCHam}
    H_{\rm TC} = \sum_v \frac{1}{4}\left(1-A_v\right) + \sum_p \frac{1}{4}\left(1-B_p\right)\quad,    
\end{equation}
with $A_v$ being the vertex terms with $X$ operators, and $B_p$ being the plaquette term with $Z$ operators.
This symmetry is enacted by the unitary operator
\[U_{\rm TC}(\theta) = e^{ i\theta H_{\rm TC}}\quad.\]
Again, we have chosen a normalization and energy shift such that the Hamiltonian terms give half-quantized charges to the individual $e$ anyon excitations ($A_v=-1$) and $m$ anyon excitations ($B_p=-1$), resulting in a $\mathbb{Z}$ spectrum in periodic boundary conditions. Therefore, we have $U_{\rm TC}(2\pi) = 1$, obtained by a `telescoping' product of toric code terms and exploiting the identities
\[\prod_v A_v = \prod_p B_p = 1\quad.
\label{eq:conditions}\]
Similar to $U(1)_{\rm SSB}$, the symmetry $U(1)_{\rm TC}$ is locality-preserving, which can be seen via its action on the algebra generators
\begin{align}\label{eq:action}
    Z_j\mapsto Z_j e^{i\frac{\theta}{2}(A_{v_j}+A_{v_j'})}\,,\, X_j\mapsto X_j e^{i\frac{\theta}{2}(B_{p_j}+B_{p_j'})}\,,
\end{align}
where $A_{v_j}$ and $A_{v_j'}$ are the neighboring vertices to bond site $j$, $B_{p_j}$ and $B_{p_j'}$ are neighboring plaquettes to bond site $j$.

We will also study the related enhanced symmetry group $U(1)_e\times U(1)_m$, defined via the vertex and plaquette terms, respectively,
\[U_e(\theta) = \exp{ i\theta \sum_v\frac{1}{4}\left(1-A_v\right)}\quad, \\
U_m(\theta) = \exp{ i\theta \sum_p \frac{1}{4}\left(1-B_p\right)}\quad,\]
normalized so that  $U_{e,m}(2\pi)=1$. Notice that $U(1)_{\rm TC}\subset U(1)_e\times U(1)_m$ is the diagonal subgroup of the larger symmetry group.

Let us first review the fractionalization property of the toric code, specifically anyon excitations and lattice 1-form symmetry operators.\footnote{We use the term \textit{lattice} 1-form symmetry to emphasize that the symmetry operators are not topological, unlike those of a 1-form symmetry in relativistic quantum field theory~\cite{Gaiotto:2014kfa}. Lattice 1-form symmetry is sometimes instead called non-topological, faithful, or non-relativistic 1-form symmetry. See Refs. \cite{Seiberg:2019vrp, Qi:2020jrf, Oh:2023bnk, Choi:2024rjm, Liu:2026tcl, Brito:2026eqi} for further discussion.} Excitations of the $A_v$ and $B_p$ operators with eigenvalue $-1$ are referred to as $e$ and $m$ anyons, respectively. From the toric code ground state, $e$ anyons are created in pairs via the $\prod_{j\in \gamma_{vv'}} Z_j$ line operators along a path $\gamma_{vv'}$ that traverses links such that the end points $v$ and $v'$ represent $e$ anyons. On the other hand, $m$ anyons are created in pairs via the $\prod_{j\in \bar{\gamma}_{pp'}} X_j$ line operator along a path $\bar{\gamma}_{pp'}$ in the dual space lattice, i.e., perpendicular to the links, such that the end points $p$ and $p'$ represent the $m$ anyons. Fusion of the $e$ and $m$ anyons represents the fermion $f$.

Closed loops $C$ and $\bar{C}$ in the lattice and dual lattice of the anyon operators with $\prod_{j\in C} Z_j$ and $\prod_{j\in \bar{C}} X_j$ are symmetries of the toric code Hamiltonian as they commute with all $A_v$ and $B_p$ operators. These form the $\mathbb{Z}^{(1),e}_2\times \mathbb{Z}^{(1),m}_2$ lattice 1-form symmetry groups of the toric code, which gives rise to a mixed `t Hooft anomaly. We may also construct a fermionic loop operator $\prod_{j\in C,j'\in \bar{C}} Z_j X_{j'}$, associated with the creation and annihilation of a pair of $f$ fermionic anyons. This operator is an anomalous lattice 1-form symmetry of the toric code~\cite{hsin2025higherformanomaliesimplyintrinsic,Feng:2025qgg,Feng:2025yge,pw12-kdjx,Kapustin:2025rhp}.

\subsection{No-go for on-siteability with semi-bounded ancilla}

Analogous to the SSB case, the lattice 1-form symmetries of the toric code are non-factorizable in the ground state algebra of the toric code since factorizing them create anyon excitations. This fractionalization effect gives us an obstruction to on-siteability with bounded ancilla.
\begin{prop}
    $U(1)_{\rm TC}$ on an $L \times L'$ torus is not on-siteable by any QCA of spread $s < L/8$, even in the presence of semi-bounded ancilla.
    \label{prop:tcrotors}
\end{prop}
\begin{proof}
    The proof is nearly the same as Proposition \ref{prop:ssbfiniteancilla}. Indeed, the ground state algebra of $H_{\rm TC}\otimes 1+1\otimes H_0$ has long string operators $\prod_{e \in \gamma} Z_e P$ and $\prod_{e \cap \gamma' \neq \varnothing} X_e P$ (with $P$ being the toric code ground state projector) which are not factorizable over a bipartition of the torus (such as $Y = [1,L/2] \times [1,L']$ and $Z = [L/2+1,L] \times [1,L']$) unless $2s \ge L/4$, because of their mutual braiding. Therefore, via Proposition~\ref{propfactorization},
    any spread-$s$ QCA which disentangles $H_{\rm TC}\otimes 1+1\otimes H_0$ must have $s \ge L/8$. Note that exchanging $L$ and $L'$ the lower bound on the spread of the disentangler is controlled by the length of the longest cycle of the torus.
\end{proof}

As in Proposition \ref{prop:ssbfiniteancilla}, the theorem holds without any assumption on $H_0$ (other than a gap), meaning that the obstruction to on-siteability of $U(1)_{\rm TC}$ is non-invertible in the context of semi-bounded ancilla. Analogous to Proposition~\ref{prop:galoisssb}, on certain system sizes we can show that the toric code symmetries cannot be on-sited via \textit{any} unitary in the presence of finite-dimensional ancilla, whose dimensions do not scale with the system size. We formulate this obstruction in Appendix~\ref{app:galois} and Proposition~\ref{prop:galoistc}.

In fact, an analogous argument follows for any string-net topologically-ordered model with non-factorizable lattice higher-form symmetries.
\begin{cor}
    All string-net models are non-on-siteable in finite depth, even in the presence of semi-bounded ancilla.
\end{cor}

\subsubsection{Obstruction via Hall conductivity}

There is another argument based on Hall conductance which proves that $U(1)_{\rm TC}$ cannot be on-siteable using finite-dimensional ancilla. 
This discussion follows Ref. \cite{Hsin:2025jot} closely. 
We first note that a $U(1)_{\rm TC}$ rotation by $2\pi$ in half of space creates a fermion line operator $\prod_{j\in C,j'\in \bar C}Z_jX_{j'}$. 
This is the $2\pi$-flux, also known as the vison, of $U(1)_{\rm TC}$. 
Next, it is known that the topological spin $\theta_v$ of the vison is related to the quantum Hall conductance $\sigma_\text{H}$ of the $U(1)$ symmetry by $\theta_v = e^{i \pi \sigma_\text{H}}$ \cite{PhysRevLett.110.046801,Kapustin:2020bwt,Cheng:2022nds}.
Since the vison of $U(1)_{\rm TC}$ is a fermion with spin 1/2, i.e., $\theta_v= -1$, we conclude that its Hall conductance is
\begin{equation}\label{eq:hall}
U(1)_{\rm TC}:~~~\sigma_\text{H} = 1 ~~\text{mod}~~2.
\end{equation}
Finally, the authors of Ref.~\cite{2019CMaPh.373..763K} show that an on-site $U(1)$ symmetry in a commuting-projector Hamiltonian with a finite-dimensional local Hilbert space cannot have nonzero Hall conductance. 
This proves that $U(1)_{\rm TC}$ is not on-siteable using finite-dimensional ancilla.

It is useful to provide a field-theoretic perspective on the $U(1)_{\rm TC}$ symmetry. We view $U(1)_\mathrm{TC}$ as a global symmetry of the toric code Hamiltonian itself. The low-energy field theory describing the ground space is the $\mathbb{Z}_2$ gauge theory coupled to a background gauge field $A$ for $U(1)_{\rm TC}$ \cite{Hsin:2025jot}:
\begin{equation}\label{eq:Z2CS}
{2\over 2\pi} adb+ {1\over 2\pi }adA +{1\over 2\pi }bdA
\end{equation}
where $a,b$ are dynamical $U(1)$ gauge fields.  

Even though the microscopic $U(1)_{\rm TC}$ acts nontrivially on the local lattice operators (see Eq. \eqref{eq:action}), it does not act on any local operator in the continuum, as there is no nontrivial gauge-invariant local operator in the $\mathbb{Z}_2$ gauge theory. Instead, the Wilson lines $e^{i \oint a}$ and $e^{i\oint b}$, corresponding to the $e$ and $m$ anyons, carry fractional $U(1)_{\rm TC}$ charges. 

To see the Hall conductance, we can locally integrate out $a$, which sets $b=-A/2$. Substituting this back into the action yields the effective response $-{1\over 4\pi }AdA$, which matches Eq. \eqref{eq:hall}. The mod 2 ambiguity reflects the freedom to add the bosonic 2+1d local counterterm ${2\over 4\pi} AdA$.

The effective odd Hall conductance in Eq. \eqref{eq:hall} is related to the 't Hooft anomaly between $U(1)_{\rm TC}$ and time-reversal (or reflection) symmetry in Ref. \cite{Gioia:2025bhl}. This is the bosonic version of the well-known ``parity" anomaly \cite{Niemi:1983rq,Redlich:1983kn,Redlich:1983dv,Alvarez-Gaume:1984zst,Witten:2016cio}. 
Before we activate the $U(1)_{\rm TC}$ background gauge field $A$, the $\mathbb{Z}_2$ gauge theory Lagrangian ${2\over 2\pi}adb$ is invariant under the time-reversal transformation $a_0\to -a_0,  a_i\to a_i,b_0\to b_0, b_i\to -b_i$. However, after coupling to the background $U(1)_{\rm TC}$ gauge field $A$, the coupled Lagrangian \eqref{eq:Z2CS} is no longer invariant under any time-reversal transformation. This shows the anomaly between $U(1)_{\rm TC}$ and time-reversal.

\subsection{On-siteability constructions with unbounded ancilla}
\label{sec:onsiteTC}

The above arguments no longer apply when we deal with infinite-dimensional unbounded ancilla. From the conjectured classification of anomalies via the bulk-boundary correspondence, since there are believed to be no $U(1)$ or $U(1) \times U(1)$ SPT phases in 3+1d, one may expect these symmetries to be on-siteable. Previous work \cite{shirley2025anomalyfreesymmetriesobstructionsgauging,Tu_2026} demonstrated that in the case of finite-dimensional ancillae, obstructions to on-siteability exist which have no corresponding SPT phase. However, since these obstructions are based on the non-triviality of the translation QCA, which becomes trivial with infinite-dimensional ancillae \cite{Jones:2026gmd}, we do not currently know of any non-on-siteable symmetry which does not have a corresponding SPT. We will make more comments on this in the discussion.

To apply Proposition~\ref{prop:disentanglerlemma}, we must be able to locally factorize the toric code Hamiltonian into mutually-commuting $O(1)$ pieces that individually maintain the same charge quantization. This turns out to be simple for the $4\pi$-periodic extension of $U(1)_e \times U(1)_m$, which allows half-charge quantization as opposed to integer-charge quantization.
\begin{prop}
    $U(1)_e \times U(1)_m$ is on-siteable in finite depth with the addition of half-charge unbounded ancillae.
    \label{prop:halfchargeTC}
\end{prop}
\begin{proof}
    Place a rotor at each vertex and plaquette with rotor angles $\phi_v,\xi_p\in\mathbb{R}/\mathbb{Z}$. We will construct an explicit circuit that simultaneously disentangles the Hamiltonians
    \begin{align}
        H_e=\sum_v\left[\frac{1}{4}(1-A_v)\otimes 1+1\otimes \frac{1}{2}N_v\right]\quad,\\
        H_m=\sum_p\left[\frac{1}{4}(1-B_p)\otimes 1+1\otimes \frac{1}{2}N_p\right]\quad,
    \end{align}
    where $N_{v}=-\frac{i}{2\pi}\frac{\partial}{\partial\phi_v}$ and $N_{p}=-\frac{i}{2\pi}\frac{\partial}{\partial\xi_p}$ is the number operator of the rotor at $v,p$. Notice that the rotor term is \textit{half-charged} due to the $\frac{1}{2}$ prefactor, rather than integer-charged. This means that the $U(1)_e\times U(1)_m$ symmetry, generated by the above Hamiltonians, is $4\pi$-periodic.

    One can check that the unitary
    \begin{align}
        W=\prod_v e^{2\pi i\phi_v \frac{1}{2}(1-A_v)}\prod_p e^{2\pi i\xi_p \frac{1}{2}(1-B_p)}\,,
    \end{align}
    is both well-defined in the rotor space and fully disentangles the above Hamiltonians, and therefore on-sites $U(1)_e \times U(1)_m$.\footnote{This construction is equivalent to applying Proposition \ref{prop:disentanglerlemma} to the symmetries
    \[U_e'(\theta) = U_e(2\theta)\,, \\
    U_m'(\theta) = U_m(2\theta)\,.\]}
\end{proof}
It follows that $U(1)_{\rm TC}\subset U(1)_e \times U(1)_m$ is also on-siteable with half-charged unbounded generators.

It is less obvious how to on-site the symmetry with integer-charged ancilla. First, we give a physically-intuitive construction for the on-siting $U(1)_{\rm TC}$ with the additional presence of integer-charged unbounded ancilla \textit{and} fermionic ancillae, inspired by the bosonization transformation of Ref.~\cite{CHEN2018234}.
\begin{prop}
    \label{prop:onsitingTC}
    $U(1)_{\rm TC}$ is on-siteable in finite depth in the presence of integer-charged unbounded ancillae and fermionic ancillae.
\end{prop}
\begin{proof}
   We add a rotor $\phi_v\in \mathbb{R}/\mathbb{Z}$ on every vertex $v$ and denote its conjugate by $N_v = - {i\over 2\pi} {\partial \over \partial \phi_v}$. 
    We will explicitly construct the disentangling circuit for the Hamiltonian
    \begin{align}
        H=H_{\rm TC}\otimes 1+\sum_{v} 1\otimes N_v\quad,
    \end{align}
    where we have added integer-charged rotors at all vertices and a chargeless complex fermion, with Majorana operators $\gamma$ and $\bar\gamma$, at each plaquette. 
    First, notice that we can rewrite $H_{\rm TC}$ as
    \begin{align}
        H_{\rm TC}=\sum_v\left[\frac{1}{4}(1-A_v B_{p(v)})+\frac{1}{4}(1-A_v)(1-B_{p(v)})\right], \label{eq:HTC_decompose}
    \end{align}
    where we choose $p(v)$ to be the north-east plaquette to the vertex $v$. The first term is locally half-integer quantized, and the second term is a projector that is integer-quantized. Denote the projector as
    \begin{align}
        P_v=\frac{1}{4}(1-A_v)(1-B_{p(v)})\,.
        \label{eq:PV}
    \end{align}
    Now, let us apply the disentangling FDQC to remove the second term via
    $$V=\prod_v e^{ 2\pi i \phi_v P_v}\,.$$
    It satisfies
    \begin{align}
        V (1\otimes N_v) V^\dag &= 1\otimes N_v-P_v\otimes 1\,,\nonumber\\
        VHV^\dag=\sum_v \frac{1}{4}&(1-A_v B_{p(v)})\otimes 1+\sum_{v}1\otimes N_v\,.
        \label{eq:VHVdag}
    \end{align}
    Now, we will modify the first term, which is locally half-integer quantized, to become integer-quantized by adding telescoping terms that eventually add to zero. Specifically, we define the following operators $\chi_{E(v),W(v),N(v),S(v)}$ at vertex $v$ as:
    \begin{figure}[H]
    \centering
    \includegraphics[width=0.8\linewidth]{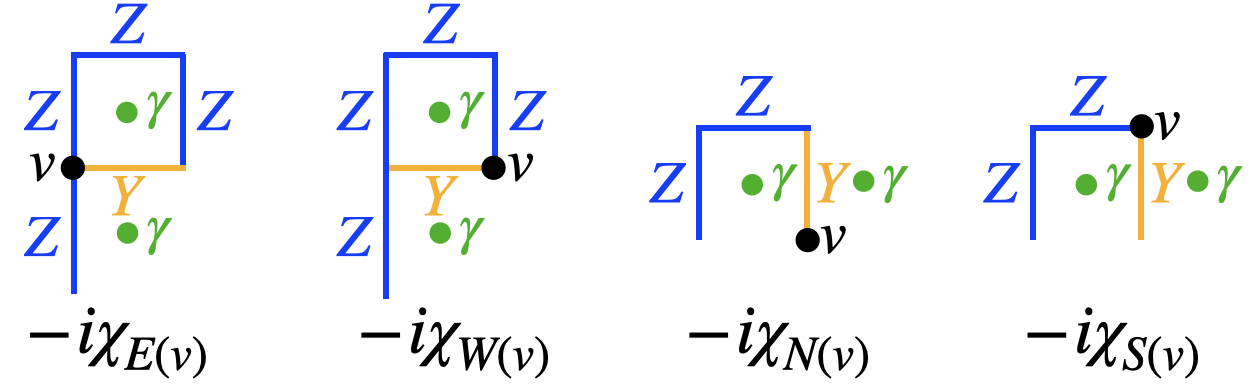}
    \label{fig:TCpieces}
    \end{figure}
    
    Note that in each plaquette of $\chi_{E(v),W(v),N(v),S(v)}$, we choose to use the $\gamma$ Majorana fermion as opposed to the second $\bar\gamma$ Majorana fermion. The ordering rule of the Majorana operators is $\gamma_{\rm up} \gamma_{\rm down}$ for $\chi_{E(v),W(v)}$ and $\gamma_{\rm right} \gamma_{\rm left}$ for $\chi_{N(v),S(v)}$. These operators represent the decomposition of $e$, $m$ anyons and $\gamma$ Majorana fermion traveling in a loop around vertex $v$, with the $e$ anyon taking a more complicated path to ensure the correct commutation relations. One can show that $\chi_{E(v)}\chi_{W(v)}\chi_{N(v)}\chi_{S(v)}=A_vB_{p(v)}$, and that $\chi_{E(v),W(v),N(v),S(v)}$ commute with each other as well as $A_{v'}B_{p(v')},\forall v'$. These operators allow us to rewrite Eq.~\ref{eq:VHVdag} as
    \begin{align}
        VHV^\dag=\sum_v O_v\otimes 1+\sum_{v}1\otimes N_v\,,
        \label{eg:TClocalfactorization}
    \end{align}
    where
    \begin{equation}
    O_v=\frac{1}{4}(1-A_v B_{p(v)}+\chi_{N(v)}+\chi_{E(v)}-\chi_{S(v)}-\chi_{W(v)}) \label{eq:OV}
    \end{equation}
    which is now integer valued, and $[O_v,O_{v'}]=0$, $\forall v,v'$. Now, define
    \begin{align}
        W=\prod_v e^{2\pi i \phi_v O_v}\,,\quad U:=WV\,,
    \end{align}
    such that
    $$UHU^\dag=\sum_{v}1\otimes N_v\quad.$$
    So we see that $U$ is a finite-depth circuit that has completely disentangled the toric code Hamiltonian, and can also be used to on-site the $U(1)_{\rm TC}$.
\end{proof}

Inspired by the disentangling procedure and local factorization of Eq.~\ref{eg:TClocalfactorization} into local integer spectrum pieces, we may construct an SRE eigenstate of $U(1)_{\rm TC}$ with purely fermionic ancillae, \textit{without} the need for rotors and in fact without the need for adding any new terms to the toric code Hamiltonian.
\begin{prop}
    \label{prop:TCSREfinitedim}
    $U(1)_{\rm TC}$ possesses an SRE eigenstate in the presence of uncharged fermionic ancillae.
\end{prop}
\begin{proof}
Note that a sufficient condition for $\ket \Psi$ to be an eigenstate of $H_{\rm TC}$ is $A_v B_{p(v)} \ket \Psi = - \ket \Psi$ ($\forall v$), where we use the same notation for $p(v)$ as in the proof of Proposition~\ref{prop:onsitingTC}. Indeed, since $A_v$ is an involution, it implies $(A_v + B_{p(v)}) \ket \Psi = 0$, so that the state is clearly an eigenstate of $\sum_v A_v + \sum_p B_p$.

To construct an SRE state $\ket \Psi$ with this property, we use the fact that $A_v B_{p(v)} = \prod_{e \ni v} \chi_e$ where $\chi_e$ was introduced in the proof of Proposition~\ref{prop:onsitingTC} with $e$ being the edge with the Pauli-$Y$ operator. It is thus sufficient to construct an SRE state $\ket \Psi$ such that for all edges $e$, we have $\chi_e \ket{\Psi} = (-1)^{s_e} \ket{\Psi}$ where the choice of signs is such that around a vertex we have $\prod_{e \ni v} (-1)^{s_e} = -1$. (We assume that the number of vertices is even for simplicity).

To obtain such a simultaneous SRE eigenstate for all $\chi_e$, first define the unitary involution $W_e = \frac{1}{\sqrt{2}} (Z_e + \chi_e)$ which satisfies $W_e Z_e W_e^\dagger = \chi_e$. Since $[W_e, W_{e'}] = 0$, we can define the finite-depth unitary circuit $W = \prod_e W_e$ such that $W Z_e W^\dagger = \chi_e$. The desired SRE state is thus simply
$$ \ket \Psi = W \; \left(\ket{ \{Z_e = (-1)^{s_e}\}} \otimes \ket{\textrm{fermionic vac}}\right). $$
\end{proof}

In fact, \textit{no} fermionic ancillae are actually required to on-site $U(1)_{\rm TC}$, or even the larger group $U(1)_e \times U(1)_m$. Here we give such a construction for on-siting $U(1)_e \times U(1)_m$ without fermionic ancilla, albeit with potentially less physical intuition.
\begin{prop}
    \label{prop:emonsite}
    $U(1)_{e}\times U(1)_m$ is on-siteable in finite depth in the presence of integer-charged unbounded ancilla.
\end{prop}
\begin{proof}
    As in Proposition~\ref{prop:halfchargeTC}, place a rotor at each vertex and plaquette with rotor angles $\phi_v,\xi_p\in\mathbb{R}/\mathbb{Z}$. We will also construct an explicit circuit that simultaneously disentangles the Hamiltonians
    \begin{align}
        H_e=\sum_v\left[\frac{1}{4}(1-A_v)\otimes 1 + 1\otimes N_v\right]\quad,\\
        H_m=\sum_p\left[\frac{1}{4}(1-B_p)\otimes 1 + 1\otimes N_p\right]\quad,
    \end{align}
    where $N_{v}=-\frac{i}{2\pi}\frac{\partial}{\partial\phi_v}$ and $N_{p}=-\frac{i}{2\pi}\frac{\partial}{\partial\xi_p}$ is the number operator of the rotor at $v,p$. In contrast to Proposition~\ref{prop:halfchargeTC}, here the rotor charge is fully integer-quantized.

    Notice that the operator
    \begin{align}
        V=\prod_v e^{2\pi i\phi_v \frac{1}{4}(1-A_v)}\prod_p e^{2\pi i\xi_p \frac{1}{4}(1-B_p)}
    \end{align}
    is almost the correct operator needed to remove the $A_v$ and $B_p$ terms, similar in spirit to the construction in Prop.~\ref{prop:onsitingTC}, with raising operators $T_v=e^{2\pi i\phi_v}$ and $T_p=e^{2\pi i\xi_p}$. However, since $\frac{1}{4}(1-A_v)$ and $\frac{1}{4}(1-B_p)$ are not individually integer quantized, unlike $P_v$ in Eq.~\ref{eq:PV}, the terms in $V$ are not invariant under gauge transformations $\phi_v\mapsto \phi_v+\alpha_v$, $\xi_p\mapsto \xi_p+\beta_p$, where $\alpha_v,\beta_p\in \mathbb{Z}$, since
    \begin{align*}
        \phi_v\mapsto \phi_v+\alpha_v: \quad V\mapsto V \prod_v(A_v)^{\alpha_v}\quad,\\
        \xi_p\mapsto \xi_p+\beta_p: \quad V\mapsto  V \prod_p (B_p)^{\beta_p}\quad.
    \end{align*}
    Although this gauge non-invariance can be cured by replacing $\phi_v \to \phi_v-\lfloor\phi_v\rfloor$ and $\xi_p \to \xi_p-\lfloor \xi_p \rfloor$, the resulting operator fails to properly disentangle the Hamiltonians due to its discontinuity causing undesirable commutation relationships with the rotor charge operator. Instead, we will construct another operator which fixes this non-invariance, while still maintaining the desirable commutation relationship with the rotor charge operators. Consider
    \begin{align}
        W=\prod_e &e^{i\pi(\phi_{T(e)}-\phi_{B(e)})\lfloor \xi_{L(e)}-\xi_{R(e)}\rfloor}\nonumber\\
        &\times X_e^{\lfloor\phi_{T(e)}-\phi_{B(e)}\rfloor}Z_e^{\lfloor \xi_{L(e)}-\xi_{R(e)}\rfloor}\quad,
    \end{align}
    where $e$ are oriented edges (horizontal point east, vertical point north) with $T(e),B(e)$ being the vertex at the tip and bottom of the edge, and $L(e),R(e)$ being the plaquettes at the left and right of the edge. The last two terms will contribute $A_v$ and $B_p$ terms upon the respective gauge transformations, and the first pure phase term is to ensure that the $X_e$ terms can be pulled out of the product to create $A_v$ with the correct sign. One can check that under the relevant shifts
    \begin{align*}
       \phi_v\mapsto \phi_v+\alpha_v: \quad W\mapsto W \prod_v (A_v)^{\alpha_v}\quad,\\
       \xi_p\mapsto \xi_p+\beta_p: \quad W\mapsto  W \prod_p (B_p)^{\beta_p}\quad.
    \end{align*}
    Crucially, the operator $W$ commutes with both $\sum_v N_v$ and $\sum_p N_p$.
    
    So the combined operator $U:=WV$ is invariant under independent integer shifts of $\phi_v$ and $\xi_p$, and is unitary since $UU^\dag=1$. It follows that
    \begin{align*}
        e^{i\theta\sum_v N_v} U e^{-i\theta \sum_v N_v}=U e^{i\theta \sum_v \frac{1}{4}(1-A_v)}\,,\\
        e^{i\theta \sum_p N_p} U e^{-i\theta \sum_p N_p}=U e^{i\theta \sum_p \frac{1}{4}(1-B_p)}\,,
    \end{align*}
    since $W$ is invariant under uniform rotor angle shifts, while $V$ transforms with the respective vertex or plaquette terms.
    This implies, via Stone's theorem on one-parameter unitary groups,
    \begin{align*}
        U H_e U^\dag&=\sum_v N_v\,,\\
        U H_m U^\dag&=\sum_p N_p\,.
    \end{align*}
    Finally, let us demonstrate that $U$ is locality-preserving, by showing that it can be written as an FDQC. Recall that $U$ is invariant under gauge transformations $ \phi_v\mapsto \phi_v+\alpha_v, \xi_p\mapsto \xi_p+\beta_p$.  
    We can therefore perform a gauge transformation to replace $\phi_v\mapsto \phi_v-\lfloor \phi_v\rfloor $ and $\xi_p\mapsto \xi_p-\lfloor \xi_p\rfloor$, so that every local factor in $U$ is now gauge-invariant. 
    Since $U$ is diagonal in the $\phi_v$, $\xi_p$ basis, it is then clear that we can write $U$ as a product of locally gauge-invariant circuit elements, forming an FDQC.
    
    This completes the proof, as we have constructed an explicit FDQC disentangler for both Hamiltonians, and are therefore able to on-site $U(1)_e\times U(1)_m$.
\end{proof}
We expect that similar methods should generalize to other topologically-ordered Hamiltonians, such as quantum-double and string-net models~\cite{KITAEV20032,PhysRevB.71.045110}.

It follows that the toric code Hamiltonian possesses SRE eigenstates that can be created using unbounded charged ancillae.
\begin{cor}
    \label{cor:SRETC}
    $U(1)_{\rm TC}$ possesses an SRE eigenstate in the presence of unbounded charged ancilla.
\end{cor}
\begin{proof}
    Define $\ket{\Psi_{\rm SRE}}=U^\dag\ket{\mathbf{0}}$, where $U$ is the FDQC given in Prop.~\ref{prop:emonsite} and $\ket{\mathbf{0}}$ is the product state of all spin up on the edges and $0$ in the rotor charge operator basis. By construction, $\ket{\Psi_{\rm SRE}}$ is an SRE eigenstate of $U(1)_{\rm TC}$ with zero energy where in the presence of these charged ancillae we now have the generator being $H_{\rm TC} + \sum_v N_v + \sum_p N_p$.
\end{proof}
This SRE state will be a complicated superposition of  excited states of the toric code and rotor charges with the rotor energy canceling out the toric code excitation energy.

\subsection{$U(1)_{\rm TC'}$}

In fact, there is a model, and associated $U(1)$ symmetry, that is a relative of $H_{\rm TC}$ (Eq.~\ref{TCHam}), which demonstrates all of the above features discussed above in a simpler way. Consider the Hamiltonian~\cite{PhysRevB.96.195150}
\begin{align}
    H_{\rm TC}' = -\sum_v \left[\left(\frac{1+A_v}{2}\right) \left( \frac{1+B_{p(v)}}{2} \right)-1\right]\,,
    \label{eq:TCprime}
\end{align}
where $p(v)$ is the north-east plaquette of vertex $v$. This model has the same ground states, and associated degeneracy, as the usual toric code Hamiltonian. It has an integer spectrum, however with a curious minimum energy gap of 2 (rather than 1) above the ground state, and also has an excitation energy of 3 (e.g., a far separated $e$ anyon line and $m$ anyon line end on the same vertex and the associated plaquette resulting in an $f$ excitation) and all the higher integers. The combination of the existence of an energy gap and topologically-ordered ground states means that this model is also topologically-ordered. We will refer to the $U(1)$ symmetry generated by $H_{\rm TC}'$ as $U(1)_{\rm TC'}$.

Another crucial difference between $U(1)_{\rm TC'}$ and $U(1)_{\rm TC}$ is the charge assignment for the anyons. In particular, $e$ and $m$ charges differ for the two models, as can be seen below:
\begingroup
\renewcommand{\arraystretch}{1.3}
\begin{table}[h!]
    \centering
    \begin{tabular}{c|c|c|c}
         & ~$e$ charge~ & ~$m$ charge~ & ~$f$ charge~ \\\hline
        $U(1)_{\rm TC}$ & $\frac{1}{2}$ & $\frac{1}{2}$ & 1\\
        $U(1)_{\rm TC'}$ & 1 & 1 & 1\\
    \end{tabular}
    \caption{Charge assignments of $e$, $m$, and $f$ anyons for $U(1)_{\rm TC}$ and $U(1)_{\rm TC'}$.}
    \label{tab:TCvsTCpcharges}
\end{table}
\endgroup

\noindent This difference in charge assignment from half integers for $U(1)_{\rm TC}$ to integers for $U(1)_{\rm TC'}$ is the essential reason as to why this model will be much simpler to analyze than the previous toric code case, despite both being topological orders. Let us now analyze the usual features of SRE eigenstates and on-siteability.

This model possesses SRE eigenstates even in the absence of ancillae.
\begin{cor}
    $U(1)_{\rm TC'}$ possesses SRE eigenstates, in particular product-states on systems with an even number of plaquettes, in the absence of ancillae.
\end{cor}
\begin{proof}
    On system sizes with an even number of plaquettes, we may simply construct a product state with $B_{p}=-1$ for all plaquettes $p$. For example if there is an even number of plaquettes in the vertical direction, then all spins living on every second row of horizontal edges are in $\ket{Z_e=-1}$, while all other edges are in $\ket{Z_e=+1}$.

    On odd system sizes with an odd number of plaquettes, we may create a state that is a product state except around one vertex. Note that since the number of plaquettes is odd, we can only require $B_p=-1$ for all plaquettes $p\neq p_0$, where $B_{p_0}=1$. For the even number of plaquettes with $B_p=-1$, we may pair neighboring plaquettes up and flip their joint spin to $\ket{Z_e=-1}$ while keeping all other spins $\ket{Z_e=+1}$ (including the ones on $p_0$), and call this configuration $\ket{\Psi_0}$. On the remaining plaquette $p_0$ and associated south-west vertex $v_0$, create the state
    \begin{align}
        \ket{\Psi}=\frac{1}{\sqrt{2}}\left(\ket{\Psi_0}-A_{v_0}\ket{\Psi_0}\right)\,.
    \end{align}
     This is an SRE eigenstate of $U(1)_{\rm TC'}$, which is also the most excited state, i.e., the ``sky" state.
\end{proof}

Following the proofs of Proposition~\ref{prop:tcrotors}, the same proof also shows that $U(1)_{\rm TC'}$ is not on-siteable in finite depth with semi-bounded ancillae.
\begin{prop}
    $U(1)_{\rm TC'}$ is non-on-siteable in finite depth, even in the presence of semi-bounded ancilla.
\end{prop}
Just like $U(1)_{\rm TC}$, its on-siteability changes when unbounded ancillae are considered. However, in this case, it is significantly easier to construct the FDQC to on-site the symmetry, since all terms in Eq.~\ref{eq:TCprime} are already integer valued.
\begin{prop}
    $U(1)_{\rm TC'}$ is on-siteable in finite depth in the presence of integer-charged unbounded ancilla.
\end{prop}
\begin{proof}
Similar to Proposition \ref{prop:ssbonsiterotor}, we add a rotor $\phi_v\in\mathbb{R}/\mathbb{Z}$ at each vertex $v$ and consider the Hamiltonian 
 \begin{align}
        H = H_{\rm TC'}\otimes 1 +\sum_v 1\otimes N_v\,.
\end{align}
It can then be disentangled by the FDQC $W=\prod_v e^{2\pi i \phi_v P_v}$ with $P_v=-\left(\frac{1+A_v}{2}\right) \left( \frac{1+B_{p(v)}}{2} \right)+1$.    
\end{proof}
Notice that the on-siting procedure was much simpler than for $U(1)_{\rm TC}$ (e.g., Proposition~\ref{prop:emonsite}) due to the charge assignment differences in Table~\ref{tab:TCvsTCpcharges}.

\section{Discussion}
\label{sec:discussion}

In this paper, we have studied on-siteability of $U(1)$ symmetries, generated by commuting-projector Hamiltonians. We prove that phases with fractionalized excitations, such as spontaneous-symmetry broken phases and topological orders, cannot be on-sited in finite-depth, even in the presence of semi-bounded ancilla. In particular, we explicitly demonstrate this for two quintessential examples: $U(1)_{\rm SSB}$, generated by the 1+1d spontaneous symmetry-broken Ising Hamiltonian, and $U(1)_{\rm TC}$, generated by the 2+1d toric code Hamiltonian. However, once infinite-dimensional ancillae with unbounded generators are included, $U(1)_{\rm SSB}$, and perhaps surprisingly, $U(1)_{\rm TC}$ are both on-siteable in finite-depth. We give explicit disentangling circuits for both the symmetries and their associated Hamiltonians.

We expect that similar disentangling circuits should hold straightforwardly for other commuting-projector models, associated to other topological orders such as quantum doubles. Immediate questions include how broad these on-siteability constructions extend to: are all commuting-projector Hamiltonians on-siteable in the presence of unbounded ancilla? A potentially easier question is whether all such Hamiltonians are invertible, i.e., can be disentangled with the presence of another Hamiltonian.

Another interesting aspect is whether there are finite-range Hamiltonians with quantized spectra which are \textit{not} commuting-projector Hamiltonians. All the studied examples are commuting-projector models. If there are such quantized-spectra non-commuting Hamiltonians, would they also be on-siteable? Discarding the finite-range condition and allowing for exponential tails, there are Hamiltonians of flat-band Chern insulators~\cite{Chen_2014}, which are known to necessarily possess non-commuting terms~\cite{2019CMaPh.373..763K}. It would be interesting to show whether such models are also disentangleable in the presence of ancillae via finite-time Hamiltonian evolution.

Our results support the idea that unbounded ancillae are necessary for on-siteability to match the classifications of anomalies in quantum field theory, and of SPTs in one dimension higher. There are also other candidates for being potential measures for lattice anomaly, such as blendability. We study this concept in Appendix~\ref{sec:blendability}, where we find some interesting differences between $U(1)_{\rm SSB}$, $U(1)_{\rm TC}$, and $U(1)_e\times U(1)_m$, given in Table~\ref{tab:TC}. In this paper, we have carefully investigated the area of $U(1)$ symmetries generated by topological Hamiltonians in the context of different types of ancilla and on-siteability, and hope that the results will aid in the understanding of anomalous symmetries on the lattice.

\section{Acknowledgements}

We are thankful for insightful and fun discussions with Xie Chen, Alexander M Czajka, Lukasz Fidkowski, Roman Geiko, Corey Jones, Alexei Kitaev, Ho Tat Lam, Zhiyao Lu, Andy Lucas,  Nikita Sopenko, Ashvin Vishwanath, and Carolyn Zhang. 
L.G. acknowledges support from the Department of Energy under Grant No. DE-SC0024324 and the Walter Burke Institute for Theoretical Physics
at Caltech and the Caltech Institute for Quantum Information and Matter.
S.D.P. acknowledges support from the Simons Collaboration on Ultra-Quantum Matter, which is a grant from the Simons Foundation (651446, XGW), as well as the Marvin L. Goldberger Membership, the William Loughlin Membership, and the IBM Einstein Fellowship Fund at the Institute for Advanced Study.
S.H.S. was supported in part by the Simons Collaboration on Ultra-Quantum Matter, which is a grant from the Simons Foundation (651444,  SHS), and in part by the U.S. Department of Energy, Office of Science, Office of High Energy Physics of U.S. Department of Energy under grant Contract Number  DE-SC0012567. R.T. is supported by the Mani L. Bhaumik Presidential Term Chair. 
Some of the derivations were assisted by ChatGPT 5.6 Sol Pro and verified by the authors.

\bibliography{main.bib}

\appendix

\onecolumngrid

\section{Connection between on-siteability of a $U(1)$ symmetry versus its generator}
\label{app:cononsite}

In this appendix, we will show that a $U(1)$ symmetry is on-siteable if and only if its generator can be disentangled to a sum of on-site terms.
\begin{prop}
    A $U(1)$ symmetry, enacted by operator $e^{i\theta H}$, is on-siteable with(out) ancillae via unitary $V$ if and only if its generator $H$ is on-siteable with(out) ancillae via unitary $V$.
    \label{prop:Stone}
\end{prop}
\begin{proof}
    First note that by Stone's theorem of one-parameter unitary groups, we may always write the unitary operator $U(\theta)$ of a $U(1)$ symmetry in terms of its \textit{unique} Hermitian generator $H$, with $U(\theta)=e^{iH\theta}$.

    Let us now prove the forward direction. Assume that $U(1)$ is on-siteable in the presence of ancillae (the same argument holds without ancillae by setting their contribution to the identity). This implies that there exists a unitary $V$ such that $\forall\theta$,
    \begin{align}
        V \left[U(\theta) \otimes U_0(\theta)\right] V^\dag=U_0' (\theta)\quad,
    \end{align}
    where $U_0(\theta)=e^{iH_0\theta}$ is a unitary (for generality, let us note it is not necessarily on-site) on ancilla degrees of freedom, which may be finite or infinite-dimensional, with a Hermitian generator $H_0$ with integer spectrum, and $U_0' (\theta)=e^{iH_0'\theta}$ is an on-site unitary with Hermitian on-site generator $H_0'$. This would imply
    \[V \left(H\otimes\mathds{1}+\mathds{1}\otimes H_0\right)V^\dag = -i V\frac{d}{d\theta}\left[U(\theta) \otimes U_0(\theta)\right]V^\dag\bigg|_{\theta=0} \\
    = -i \frac{d}{d\theta}\left[V(U(\theta) \otimes U_0(\theta))V^\dag\right]\bigg|_{\theta=0} = H_0'\quad,\]
    where in the first and final equalities we use Stone's theorem on one-parameter unitary groups, and the second equality follows because $V$ is continuous as a bounded linear map on Hilbert space and independent of $\theta$, with the derivatives signifying Hilbert space norm derivatives on the domains of the generator.

    To prove the other direction, let us now assume that the generator $H$ of a $U(1)$ symmetry, can be made on-site via unitary $V$
    \begin{align}
        V \left(H\otimes\mathds{1}+\mathds{1}\otimes H_0\right)V^\dag= H_0'\,,
    \end{align}
    where $H_0$ is the integer-quantized ancillae generator (which, for generality, may or may not be on-site), and $H_0'$ is the integer-quantized generator on the entire system. By matrix exponentiation, we have 
    \begin{align}
        V e^{i\theta \left(H\otimes\mathds{1}+\mathds{1}\otimes H_0\right)}V^\dag = e^{i\theta V \left(H\otimes\mathds{1}+\mathds{1}\otimes H_0\right)V^\dag}= e^{i\theta H_0'}\,,
    \end{align}    
    so we see that the $U(1)$ symmetry is on-siteable.
\end{proof}

\section{Galois theory of spectral polynomials}
\label{app:galois}

Let $H$ be an integer-quantized Hamiltonian on a finite-dimensional Hilbert space and assume $H \ge 0$, so the energy levels are non-negative integers. Let $D_n$ be the number of eigenstates of $H$ of energy $n$. Consider the ``spectral polynomial''
\[p_H(x) = \sum_{n \ge 0} D_n x^n.\]
This polynomial is invariant under conjugation $H \mapsto U H U^\dagger$ for arbitrary unitaries, regardless of their depth. Furthermore, for $H \otimes 1 + 1 \otimes H'$, where $H'$ is also an integer-quantized (not necessarily on-site) Hamiltonian on the Hilbert space of finite-dimensional ancillae, the resulting spectral polynomial is
\[p_H(x) p_{H'}(x).\]
Note also that $p_H(1)$ is the dimension of the Hilbert space. Thus, if $H \otimes 1 + 1 \otimes H'$ is on-siteable into a tensor product operator with on-site Hilbert spaces of dimension $d_i$, $i=1,\ldots,N$, then
\[p_H(x) p_{H'}(x) = \prod_{i=1}^{N} p_i(x)\]
where $p_i(x)$ are the spectral polynomials of the local tensor factors, and in particular $p_i(1) = d_i$.

For the toric code on a cell complex with $N_v$ vertices and $N_p$ plaquettes, normalized to have integer level spacing, we have
\[\label{eqntoriccodepolynomial}p_{\rm TC}(x) = D_0 \left(\sum_{n=0}^{\lfloor N_v/2 \rfloor } {N_v \choose 2n} x^n\right)\left(\sum_{m=0}^{\lfloor N_p/2 \rfloor } {N_p \choose 2m} x^m\right).\]
The two factors represent the $e$ (vertex) and $m$ (plaquette) excitations, respectively, and $D_0$ is the ground state degeneracy, which is an overall topological degeneracy present in every energy level. 

\begin{prop}
\label{prop:galoistc}
Consider the toric code on an $L\times L$ square lattice with periodic boundary conditions, where $L = 2^n$ with $n\in\mathbb{Z}^+$. We have $N_v = 2^{2n}$ and $N_p = 2^{2n}$, so the factorization \eqref{eqntoriccodepolynomial} is the (unique!) factorization of $p_{\rm TC}(x)$ into irreducible polynomials. Therefore, in order to make $p_{\rm TC}(x) p_{H'}(x)$ on-site by a circuit of any depth, it will be necessary to have a site with spectral polynomial divisible by $Q_{2^{2n}}(x)$, given by
\[Q_N(x) =\sum_{n=0}^{\lfloor N/2 \rfloor } {N \choose 2n} x^n,\]
so that the dimension of this site would be a nonzero integer multiple of $Q_{2^{2n}}(1)=2^{L^2-1}$.
\end{prop}

\begin{proof}
This polynomial $Q_N(x)$ can be rewritten
\[Q_N(x) = \frac12\left((1+\sqrt{x})^N+(1-\sqrt{x})^N\right).\]
Taking $x = t^2$, we can solve for its roots, which satisfy
\[\left(\frac{1+t}{1-t}\right)^N=-1,\]
so
\[\frac{1+t}{1-t} = e^{\pi i m/N} \qquad m\text{ odd}, 1 \leq m < N\]
or
\[x = -\tan^2 \left(\frac{\pi m}{2N}\right)\qquad m\text{ odd}, 1 \leq m < N\]
From this expression, we can see that the splitting field is the maximal real subfield of the cyclotomic field $\mathbb{Q}(e^{\pi i/N}) \cap \mathbb{R}$. The Galois group of this field is $\mathbb{Z}_{2N}^\times/\mathbb{Z}_2$ where the quotient is by the complex conjugation action. This can be used to constrain the factorization of $Q_N(x)$, since any factorization must separate the roots of $Q_N(x)$ into disjoint orbits of the Galois action.

For the case at hand $\mathbb{Z}_{2N}^\times/\mathbb{Z}_2 = \mathbb{Z}_{2^{2n-1}}$. In particular, this means that $Q_{2^{2n}}$ is irreducible. Thus, one of the $p_i$'s must be divisible by $Q_{2^{2n}}$. We therefore have $p_i(1) = Q_{2^{2n}}(1) p_i'(1)$ where $p_i'(1)$ is a nonzero integer. (If sites are composed of toric code site blocked with an ancilla, then corresponding ancilla must have dimension $\frac{1}{2}p_i(1)$.)
\end{proof}

An analogous statement follows for $H_{\rm SSB}$ with the simpler spectral polynomial
\[\label{eqnssbpolynomial}p_{\rm SSB}(x) = 2\left(\sum_{n=0}^{\lfloor L/2 \rfloor } {L \choose 2n} x^n\right),\]
where $L$ is the system size.
\begin{prop}
    Consider $H_{\rm SSB}$ on a lattice of size $L = 2^n$ with $n\in\mathbb{Z}^+$. The factorization \eqref{eqnssbpolynomial} is the (unique!) factorization of $p_{\rm SSB}(x)$ into irreducible polynomials. Therefore, in order to make $p_{\rm SSB}(x) p_{H'}(x)$ on-site by a circuit of any depth, it will be necessary to have a site with spectral polynomial $Q_{2^{n}}(x)$, which would be an output site with spectral polynomial divisible by a nonzero integer multiple of $Q_{2^{n}}(1)=2^{L-1}$.
\end{prop}
\begin{proof}
    The proof follows the proof of Proposition~\ref{prop:galoistc} step by step.
\end{proof}

\section{Blendability}
\label{sec:blendability}

\begin{table*}[t]
\centering
\renewcommand{\arraystretch}{1.15}
\setlength{\tabcolsep}{6pt}

\begin{tabular*}{\textwidth}{
    @{\hspace{12pt}\extracolsep{\fill}}
    l c c c
    @{\extracolsep{0pt}\hspace{12pt}}
}
\toprule
\textbf{Property}
    & $U(1)_{\mathrm{SSB}}$
    & $U(1)_{\mathrm{TC}}$
    & $U(1)_{\mathrm{e}}\times U(1)_{\mathrm{m}}$ \\
\midrule

On-siteability with semi-bounded ancillae
    & \makecell[c]{
        \textbf{No} \\[-1pt]
        {\footnotesize Prop.~\ref{prop:ssbfiniteancilla}}
      }
    & \makecell[c]{
        \textbf{No} \\[-1pt]
        {\footnotesize Prop.~\ref{prop:tcrotors}}
      }
    & \makecell[c]{
        \textbf{No} \\[-1pt]
        {\footnotesize Prop.~\ref{prop:tcrotors}}
      } \\

\addlinespace[0.6em]

On-siteability with unbounded ancillae
    & \makecell[c]{
        \textbf{Yes} \\[-1pt]
        {\footnotesize Prop.~\ref{prop:ssbonsiterotor}}
      }
    & \makecell[c]{
        \textbf{Yes} \\[-1pt]
        {\footnotesize
         Props.~\ref{prop:onsitingTC} and~\ref{prop:emonsite}}
      }
    & \makecell[c]{
        \textbf{Yes} \\[-1pt]
        {\footnotesize Prop.~\ref{prop:emonsite}}
      } \\

\midrule

Blendability with semi-bounded ancillae
    & \makecell[c]{
        \textbf{Yes} \\[-1pt]
        {\footnotesize Prop.~\ref{prop:SSBblend}}
      }
    & \makecell[c]{
        \textbf{Yes}\textsuperscript{*} \\[-1pt]
        {\footnotesize Prop.~\ref{prop:TCblendfermionancilla}}
      }
    & \makecell[c]{
        \textbf{No} \\[-1pt]
        {\footnotesize Sec.~\ref{sec:tcfiniteancillablend}}
      } \\

\addlinespace[0.6em]

Blendability with unbounded ancillae
    & \makecell[c]{
        \textbf{Yes} \\[-1pt]
        {\footnotesize\strut}
      }
    & \makecell[c]{
        \textbf{Yes} \\[-1pt]
        {\footnotesize Prop.~\ref{propemblendwithrotors}}
      }
    & \makecell[c]{
        \textbf{Yes} \\[-1pt]
        {\footnotesize Prop.~\ref{propemblendwithrotors}}
      } \\

\bottomrule
\end{tabular*}

\par\smallskip
\parbox{\dimexpr\textwidth-24pt\relax}{
    \footnotesize
    \textsuperscript{*} We demonstrate an intuitive construction
    with \textit{fermionic} ancillae.
}

\caption{On-siteability and blendability with semi-bounded or unbounded
ancillae for $U(1)_{\mathrm{SSB}}$, $U(1)_{\mathrm{TC}}$, and
$U(1)_{\mathrm{e}}\times U(1)_{\mathrm{m}}$.}
\label{tab:TC}
\end{table*}

By Proposition \ref{prop:disentanglerlemma}, on-siteability requires local factorizability, meaning that the symmetry representation must be expressible as a product of commuting local symmetry representations. However, there is a weaker condition called \textit{blendability} that determines whether a symmetry is able to be localized to a given region $R$ while maintaining the symmetry representation, such that it acts in the usual way deep in the interior of $R$, acts as the identity far outside $R$, and is only modified in a neighborhood of the boundary $\partial R$. When $R$ is a half-space, we call this a blend, after 
\cite{czajka2025anomalieslatticehomotopyquantum,Tu_2026}. More formally:

\begin{defin}{\rm (Blendability)}
    Let $G$ be a group, $\{U(g)\}_{g \in G}$ a set of locality-preserving unitaries on $\mathbb{Z}^d$ indexed by $G$, satisfying $U(g) U(h) = U(gh)$. We call this a locality-preserving unitary $G$ action. A blend from $\{U(g)\}_{g \in G}$ to the identity action is another locality-preserving unitary $G$ action $\{V(g)\}_{g \in G}$ such that there exists an $r$ such that the following hold:
    \begin{enumerate}
        \item For all operators $\mathcal{O}$ supported in $(-\infty,-r) \times \mathbb{Z}^{d-1}$,
        \[V(g) \mathcal{O} V(g)^\dagger = U(g) \mathcal{O} U(g)^\dagger\]
        \item For all operators $\mathcal{O}$ supported in $(r,\infty) \times \mathbb{Z}^{d-1}$,
        \[V(g) \mathcal{O} V(g)^\dagger = \mathcal{O}.\]
    \end{enumerate}
    We allow the addition of ancillae in $[-r,r] \times \mathbb{Z}^{d-1}$ on which $\{V(g)\}_{g \in G}$ may act. When $\{U(g)\}_{g \in G}$ admits a blend to the identity we say it is blendable.
\end{defin}

As mentioned above, if $\{U(g)\}_{g \in G}$ is on-siteable, then it is also blendable (with suitable ancillae added near the blend region $(-r,r) \times \mathbb{Z}^{d-1}$). Indeed, being on-siteable means there exists a circuit $W$ such that
\[U(g) \otimes S(g) = W^\dagger S(g)'W\]
(see \eqref{eqnonsitable}) where $\{S(g)\}_{g \in G}$ and $\{S(g)'\}_{g \in G}$ denote on-site $G$ actions. We can form a blend from $\{U_g \otimes S(g)\}_{g \in G}$ to the identity by restricting $S(g)'$ to a half-space and then applying $W$. We can then also discard the $S(g)$ factors, which are also on-site, which lie outside of the spread of $W$, to obtain a blend from $U(g)$ to the identity.

In one dimension, it was shown in \cite{seifnashri2025disentanglinganomalyfreesymmetriesquantum} that for finite $G$, blendability implies on-siteability (see also \cite{bols2025classificationlocalitypreservingsymmetries}). In two dimensions and higher, it is not currently known whether this is true or not. At the moment, blendability seems to be a more well-behaved property than on-siteability. For example, every locality-preserving unitary $G$-action $U(g)$ can be ``folded'' to produce an action $U(g) \otimes U(g)^{\rm rev}$ where $U(g)^{\rm rev}$ is reversed along some axis, and this folded action admits a blend. Related to this invertibility of the blend anomaly, homotopy-theoretic lattice anomalies most naturally obstruct blendability \cite{czajka2025anomalieslatticehomotopyquantum}.

To gain some intuition for blendability, let us first study the case of $U(1)_{\rm SSB}$. It is simple to see that $U(1)_{\rm SSB}$ is blendable to the identity.
\begin{prop}
    $U(1)_{\rm SSB}$ is blendable (to the identity), even in the absence of ancilla.
    \label{prop:SSBblend}
\end{prop}
\begin{proof}
    For a boundary at site at $j=1$ (left being $U(1)_{\rm SSB}$, right being the identity), we can simply choose the generator of the blend $H_{\rm blend}$ to be
    \begin{align}
        H_{\rm blend}=\sum_{j=-\infty}^0 \frac{1}{4}\left(1-Z_j Z_{j+1}\right)+ \frac{1}{4}\left(1- Z_{1}\right)\quad.
    \end{align}
\end{proof}
We observe that instead of assigning fractional energy to the $10$ and $01$ domain walls, the extra term allows us to consistently assign integer unit energy to either $10$ or $01$ and zero energy to the opposite. This possibility represents a trivial gapped boundary in which the condensation of the domain wall at the boundary is avoided.

The story is not as simple for the toric code symmetries due to the presence of fractionalized excitations $e$, $m$, and $f$, which have non-trivial braiding relations. To construct an appropriate blend, we need to first examine the boundary physics of the toric code, where we will, perhaps surprisingly, discover a difference between $U(1)_{\rm TC}$ and $U(1)_e\times U(1)_m$. The results are presented in Table~\ref{tab:TC}, and in the following sections we will give the corresponding proofs.

\begin{figure*}
    \centering
    \includegraphics[width=0.7\linewidth]{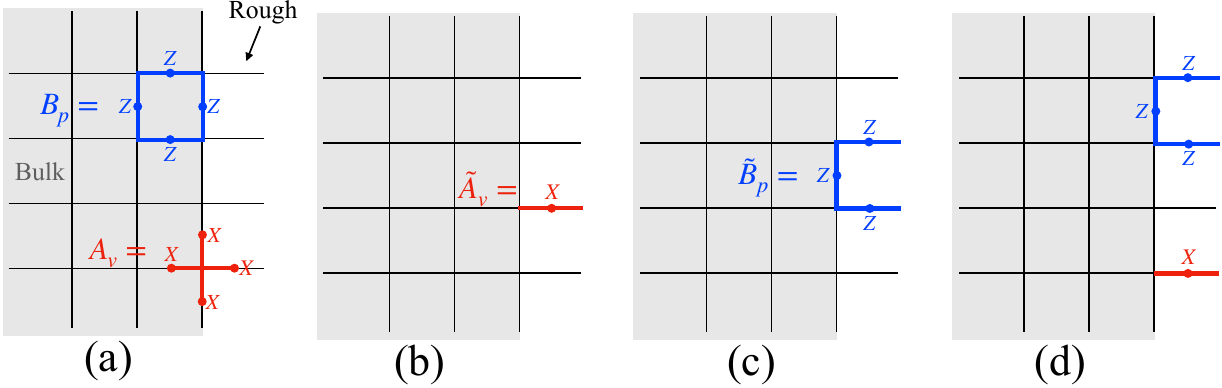}
    \caption{Possible boundary operators in the absence of ancilla degrees of freedom for a rough boundary. (a) Only toric code terms fully contained within the boundary are kept. In this case both $e$ and $m$ anyons are condensed at the boundary. (b) In addition to (a) we may add $\tilde{A}_v$ which causes $e$ to be an excitation at the boundary, however $m$ is still condensed. (c) Instead, to (a) we may add $\tilde{B}_p$ which causes $m$ to be a boundary excitation, however with $e$ still condensed. (d) If we add both $\tilde{A}_v$ and $\tilde{B}_p$, the boundary is governed by the 1+1d transverse Ising model which generically does not feature a quantized energy spectrum.}
    \label{fig:tcboundary}
\end{figure*}

\subsection{$U(1)_{\rm TC}$}

In the absence of any ancilla degrees of freedom, blends of $U(1)_{\rm TC}$ are closely related to boundary conditions for the toric code Hamiltonian. Let us briefly review some common boundary conditions. None of them will constitute a blend. We depict all four possibilities on a \textit{rough} boundary truncation for the half-space (where the vertices straddle the boundary~\cite{KitaevKong2012Boundaries}), in Fig.~\ref{fig:tcboundary}.
\begin{enumerate}
    \item Free boundary: construct $H_{\rm blend}$ by throwing away all terms that have any support outside of the half space. Such a truncation results in a Hamiltonian with extensive ground state degeneracy, as well as a half-integer spectrum, since single $e$ or $m$ anyons can be created at the boundary by local operators (and hence we say they are both ``condensed''). Thus, it does not constitute a blend.
    \item $m$-condensing boundary: We can add further boundary terms $\frac14 \sum_{v\in \partial X}(1-\tilde{A}_v)$ where $\tilde{A}_v$ are ``chopped'' vertex terms to $H_{\rm blend}$. With these terms, we can no longer excite a single $e$ anyon at the boundary without also exciting one of these terms, and thus the $e$ anyons are \emph{not} condensed at this boundary. However, $m$ remains condensed, and so there is still a half-integer spectrum. See Fig.~\ref{fig:tcboundary}~(b).
    \item $e$-condensing boundary: We could instead add boundary terms $\sum_{p\in\partial X}\tilde{B}_p$, which would avoid $m$ condensation. In this case, $e$ remains condensed and gives rise to a half-integer spectrum. See Fig.~\ref{fig:tcboundary}~(c).
    \item Ising boundary: we could add both terms: $H_{\rm Ising}=\sum_{j\in\partial X}(h\tilde{A}_{v(j)}+J \tilde{B}_{p(j)})$, but $\tilde{A}_{v(j)}$ and $\tilde{B}_{p(j)}$ do not commute and give rise to a 1+1d transverse Ising model which does not have an integer spectrum. See Fig.~\ref{fig:tcboundary}~(d).
\end{enumerate}
One feature with all the above boundary conditions is that $e^{i2\pi H_{\rm blend}}\neq 1$. (1) $U_{\rm blend}(2\pi)=\prod f$ (fermionic line), (2) $U_{\rm blend}(2\pi)=\prod m$, (3) $U_{\rm blend}(2\pi)=\prod e$, and (4) $U_{\rm blend}(2\pi)=e^{i\frac{\pi}{2} (1-H_{\rm Ising})}$. This means that they all break the requirements for being a blend of the $U(1)_{\rm TC}$ symmetry.

There is one gapped boundary which constitutes a blend, in the presence of fermionic ancillae:
\begin{prop}
    $U(1)_{\rm TC}$ is blendable to the identity in the presence of finite-dimensional fermionic ancilla.
    \label{prop:TCblendfermionancilla}
\end{prop}
\textit{Proof.} To create a blend for $U(1)_{\rm TC}$, we modify the boundary in Fig.~\ref{fig:tcboundary}a, in the following manner: finite-dimensional fermionic ancilla are added to each horizontal boundary site, denoted by Majorana creation operators $a_j$ and $b_j$. The boundary operator $A_\nu$ is modified to $\tilde{A}_{\nu}$ (see figure below), where horizontal (perpendical to the boundary) bond $X$ is discarded and replaced by $a$ and $b$ Majoranas at the same site. We also add a boundary term $i Z_{j+\frac{1}{2}} a_{j+1} b_{j}$, which allows us to confine the boundary $m$ anyon.
\begin{figure}[h]
    \centering
    \includegraphics[width=0.25\linewidth]{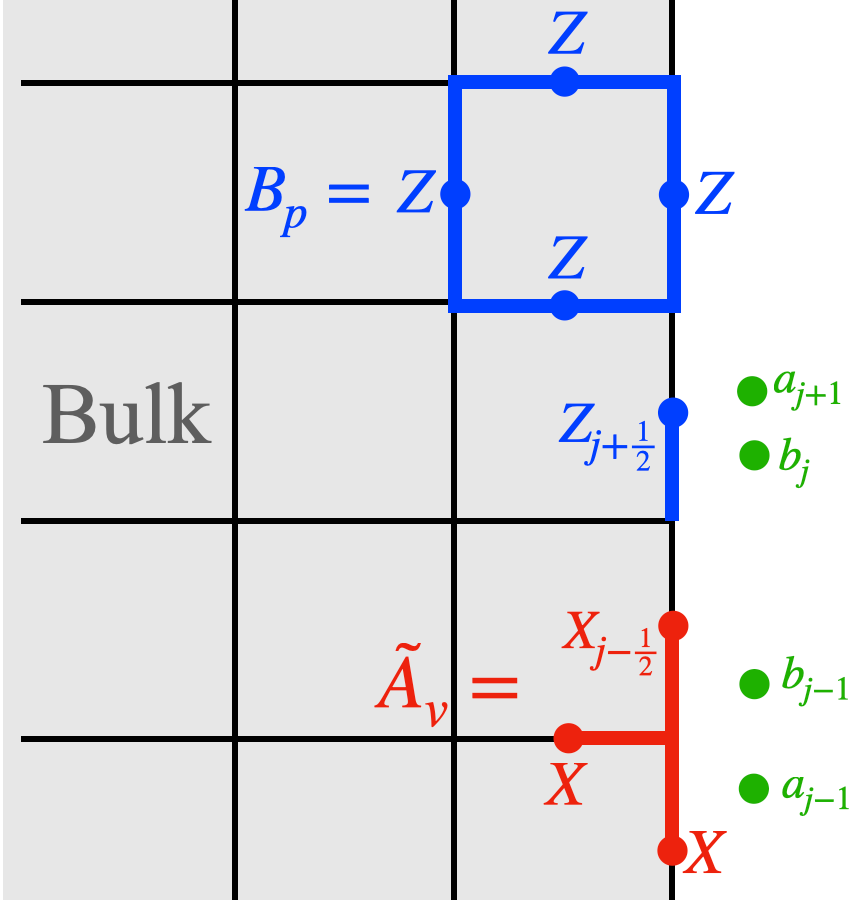}
\end{figure}
Such a construction realizes the $f$-condensed boundary with confined $e$ and $m$ anyons, thereby retaining the $\mathbb{Z}$ quantization of the symmetry. 
\qed

Despite the success in creating a blend for $U(1)_{\rm TC}$, in the next section we will give an argument for why $U(1)_e \times U(1)_m$ does not admit a blend under these conditions. We will also show that $U(1)_{\rm TC}$ does admit a blend with bosonic unbounded ancilla.

\subsection{$U(1)_e\times U(1)_m$}

Now we consider the symmetry $G=U(1)_e \times U(1)_m$. As described above (see Fig. \ref{fig:tcboundary}), $U(1)_e$ and $U(1)_m$ individually admit blends, constructed from the usual gapped boundaries of the toric code. However, these blend generators do not commute with one another, so these do not constitute a blend for $U(1)_e \times U(1)_m$ since the blend generators do not satisfy the commutative $G$ group law. In fact we will argue that, unlike the SSB and toric code Hamiltonians, $U(1)_e \times U(1)_m$ is not blendable to the identity, even in the presence of semi-bounded ancilla. We will also explicitly construct a blend to the identity for $U(1)_e\times U(1)_m$ in the presence of unbounded (bosonic) ancilla. This also gives a blend for $U(1)_{\rm TC}$.

\subsubsection{Argument for blendability obstruction with semi-bounded ancilla}
\label{sec:tcfiniteancillablend}

To tackle this problem, we will employ the method of boundary algebras~\cite{Jones_Naaijkens_Penneys_Wallick_Izumi_2025}. This is a relative of the ground state algebra we discussed above, but in the presence of a boundary. Precise definitions may be found in \cite{Jones_Naaijkens_Penneys_Wallick_Izumi_2025}. In short, given a commuting-projector Hamiltonian, which for each region $R$ has an associated ground state projector $P_R$ consisting of projections supported inside $R$, we can define a net of algebra $\mathcal{B}_R = P_{R^{+s}} \mathcal{A}_R P_{R^{+s}}$ where $\mathcal{A}_R$ denotes all the local operators supported inside $R$, and $s$ is a uniform constant. For models like the toric code which satisfy local topological order conditions, for all regions $R$ some fixed distance into the bulk, $\mathcal{B}_R \cong \mathbb{C}$. Therefore we can consider the algebra $\mathcal{B}_R$ as a net of algebras supported on the boundary of model. For the toric code with its open boundary condition, this net of algebras is bounded-spread isomorphic to the algebra of $\mathbb{Z}_2$-symmetric operators of a spin-$1/2$ chain \cite{Jones_Naaijkens_Penneys_Wallick_Izumi_2025}.

We can turn the existence of a blend into the existence of a certain operator in the boundary algebra of $H$, where the boundary algebra is the ground state algebra of $H$ truncated to a region with boundary \cite{Jones_Naaijkens_Penneys_Wallick_Izumi_2025}.

\begin{lemma}\label{propboundaryalgblend}
    Suppose $H$ is a commuting-projector Hamiltonian with integer spectrum satisfying local topological order conditions, with range $r$ terms, and $H$ admits a blend $H_{\rm blend} = H + H_\partial$ in some region $R$, where $H_\partial$ is supported in $\partial R^{+s}$. Then we obtain a locally-generated operator $\tilde H_\partial$ in the boundary algebra of $H$ satisfying
    \[e^{2\pi i \tilde H_\partial}=P_\partial e^{-2\pi i H} P_\partial\]
\end{lemma}
\begin{proof}
    By including more terms in $H$ inside $H_\partial$ and thus making the blend region larger, we may assume  $H_\partial$ commutes with each term of $H$ and hence with the boundary projector $P_\partial$. We have
    \[P_\partial e^{2\pi i H_\partial}P_\partial = e^{2\pi i \sum P_\partial h_\partial P_\partial} = P_\partial e^{-2\pi i H}P_\partial\]
    Define
    \[\tilde H_\partial = \sum P_\partial h_\partial P_\partial \]
    and the result follows.
\end{proof}

Let us now give a sketch of the proof for why $U(1)_e \times U(1)_m$ is not blendable to the identity, even in the presence of semi-bounded ancilla. Let us assume that a blend of $U(1)_e \times U(1)_m$ is possible, and then arrive at a contradiction. This assumption will allow us to construct a local (finite-range) unitary that is charged under a $U(1)$ symmetry, whose generator is semi-bounded. This semi-bounded property implies that such an operator is not possible in systems, thereby arriving at a contradiction.

    First, we note that it is sufficient to add the ancillae only near the blend region, since bulk ancillae act in an on-site fashion and can be discarded. Therefore, by Lemma \ref{propboundaryalgblend}, our blend will give rise to a $\tilde H_{\partial,e,m}$ in the boundary algebra of the toric code, potentially tensored with ancillae. We will assume these generators are bounded from below. Moreover, to be blends, they must satisfy
    \[ [\tilde H_{\partial,e},\tilde H_{\partial,m}] = 0\]
    and
    \[ \exp( 2\pi i \tilde H_{\partial,e}) = W_{\partial,m} \\
    \exp( 2\pi i \tilde H_{\partial,m}) = W_{\partial,e}\]
    where $W_{\partial,e,m}$ are boundary line operators for $e$ or $m$.

    As described above, the toric code boundary algebra (with ancillae) can be identified as the $\mathbb{Z}_2$-symmetric subalgebra of the algebra of a spin-$1/2$ chain (tensor with ancillae). We can map this to a usual 1+1d algebra by introducing a second, reference boundary. We choose this reference boundary to be a rough boundary a fixed distance from our blend region, separated by a few bulk projectors, and take the $m$-condensed gapped boundary in Fig. \ref{fig:tcboundary}(b). The ground state algebra of this configuration (call this the ``slab algebra'') is identified with the whole algebra of the spin-1/2 chain, where the missing $\mathbb{Z}_2$-odd operator is supplied by the $m$ string stretching from the reference boundary to the blend boundary.
    
    The boundary algebra includes into the slab algebra and so we will consider $\tilde H_{\partial,e}$, $\tilde H_{\partial,m}$ as elements of the slab algebra. Note that in the slab algebra, $W_{\partial,m} = 1$ since the $m$ anyon is condensed at the reference boundary, meaning if $P_S$ is the groundstate projector for the slab, $P_S W_{\partial,m} P_S = P_S$.

    Since it is local in the slab algebra, let us consider truncating $\tilde H_{\partial,e}$ to a finite interval $I$ along the boundary. $\exp(2\pi i \tilde H_{\partial,e}^I)$, which is a truncation of a locality-preserving operator, must be equal to its bulk action $W_{\partial,m}$ times two unitary operators supported near the end points of $I$. The bulk action creates an $m$ line connecting the two unitary boundary actions, which can be moved in the slab algebra to the reference boundary. This lets us write
    \[\exp(2\pi i \tilde H_{\partial,e}^I) = U_L U_R\]
    Comparing with $e^{2\pi i \tilde H_{\partial,m}}=W_{\partial,e}$, the braiding implies that the end point unitaries $U_L$, $U_R$ carry half-integer charge under $\tilde H_{\partial,e}$.

    By choosing a truncation which commutes with $\tilde H_{\partial,m}$ (for example, by Haar averaging), $U_L U_R$ is neutral. It follows that $U_L$, $U_R$ have definite, non-zero charge under $\tilde H_{\partial,e}$. However, this contradicts the boundedness of $\tilde H_{\partial,e}$, since we may use $U_L$ or $U_R$ to lower the charge arbitrarily.

\subsubsection{Unbounded ancilla}

The story changes when unbounded ancillae are added, especially since $U(1)_A\times U(1)_V$ charge pumps are possible in this context~\cite{fidkowski2025noninvertiblebosonicchiralsymmetry,thorngren2026chirallatticegaugetheories}. Using this fact, one can show that the symmetry is blendable.
\begin{prop}\label{propemblendwithrotors}
$U(1)_e \times U(1)_m$ is blendable in the presence of unbounded ancillae.    
\end{prop}
\begin{proof}
We recall the construction of the $U(1)_V \times U(1)_A$ symmetry in a chain of 1-periodic rotors $\phi_j \in \mathbb{R}/\mathbb{Z}$
with conjugate momentum $N_j = -\frac{i}{2\pi}\frac{d}{d\phi_j}$ \cite{fidkowski2025noninvertiblebosonicchiralsymmetry,thorngren2026chirallatticegaugetheories}. Let $[x]$ denote the nearest integer function $[x]=\lfloor x+\frac{1}{2}\rfloor$, where $\lfloor...\rfloor$ is the floor function.

Consider the topological unitary operator which measures the winding number of $\phi_j$
\[U_A(\alpha) = e^{ i \alpha \sum_j (\phi_{j+1}-\phi_j-[\phi_{j+1}-\phi_j])}\quad.\]
The local quantity in the sum is invariant under integer shifts $\phi_j \mapsto \phi_j + n_j$, so each local quantity also defines a unitary operator. They are commuting for different $j$ so this can be regarded as a depth-2 circuit for all $\alpha$. For $\alpha = 2\pi$, the integer part $[\phi_{j+1}-\phi_j]$ drops out and the rest of the terms become individually shift invariant and telescope to give $U_A(2\pi)=1$. Furthermore, $U_A(\alpha)$ commutes with
\[U_V(\beta) = e^{ i \beta \sum_j N_j}\quad,\]
which does a constant $\beta$ shift on the variable $\phi_j \mapsto \phi_j + \frac{\beta}{2\pi}$ and is also $2\pi$ periodic. With open boundaries, these unitaries can be observed to pump unit charge for one another. This is most easily seen by truncating $U_A^I(2\pi)=e^{-2\pi i  \phi_L}e^{2\pi i \phi_R}$ to a region $I=[L,R]$, where the endpoints carry integer $U_V(\beta)$ charge.

We are interested in a $\mathbb{Z}_2$-gauged version of this. Introduce qubits $X_{j +1/2}$, $Z_{j+1/2}$ at half-integer sites. We will think of $a_{j+1/2} = \frac12 (1-Z_{j+1/2})$ as a $\mathbb{Z}_2$ gauge field. Consider
\[\tilde U_A(\alpha) = e^{ i \alpha \sum_j \big( \phi_{j+1}-\phi_j+ \frac12 a_{j+1/2}-[\phi_{j+1}-\phi_j+\frac12 a_{j+1/2}]\big)},\label{eq:UAtilde}\]
The local terms are invariant under integer shifts of $\phi_j$ as above so this is a circuit. Furthermore, under a ``gauge transformation''
\[\phi_j \mapsto \phi_j+\frac12 \\
Z_{j\pm 1/2} \mapsto - Z_{j \pm 1/2}\]
$a_{j+1/2}-\phi_j$ either shifts by 0 or 1 and cancels the shift of the term in brackets. Thus $\tilde U_A(\alpha)$ is gauge invariant. Furthermore,
\[\tilde U_A(2\pi) = e^{ i \pi \sum_j a_{j+1/2}} = \prod_j Z_{j+1/2}\quad,\]
and $\tilde U_A(\alpha)$ commutes with $U_V(\beta)$ as defined above.

\begin{figure}[h]
    \centering
    \includegraphics[width=0.25\linewidth]{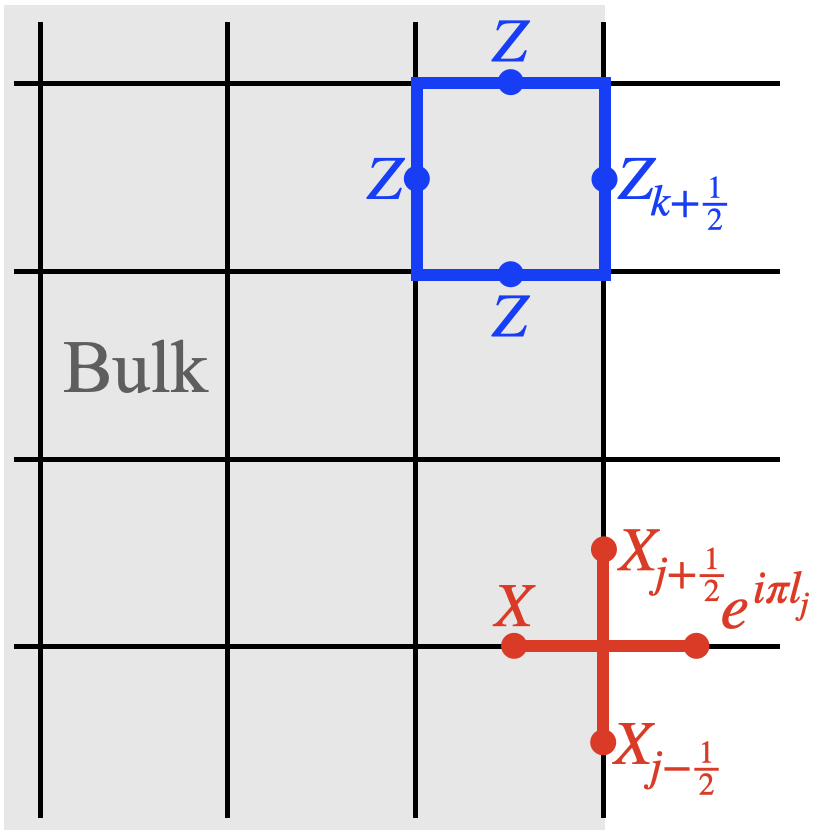}
    \caption{Rough boundary of the toric code with modified star operator (red) where the vertical links are embedded into rotors.}
    \label{fig:blendopsTC}
\end{figure}

Consider toric code with a boundary, as depicted above, where we label the horizontal (vertical) boundary edges by integers $j$ and the parallel boundary edges by half-integers $j+1/2$. Replace the qubits on the horizontal boundary edges with rotors, and modify the star terms there by replacing $X_j$ with $(-1)^{N_j}$ which acts as a $1/2$-shift $\phi_j \mapsto \phi_j +\frac{1}{2}$. The modified star terms are shown in red in Fig.~\ref{fig:blendopsTC}. These star terms act precisely as the gauge transformations below, so $\tilde U_A(\alpha)$ commutes with each of them. $\tilde U_A(\alpha)$ also commutes with the plaquette terms because those only involve $Z_{j+1/2}$. $U_V(\beta)$ also commutes with both the star and plaquette terms since it only involves $N_j$. Therefore, we may consider
\[\tilde{U}_e(\theta) = e^{i\theta \sum_v \frac{1}{4}(1-A_v)} U_V(\theta/2) \\
\tilde{U}_m(\theta) = e^{i\theta \sum_p \frac{1}{4}(1-B_p)} \tilde U_A(\theta)\]
where $A_v$ are understood to be the star terms which fit in the rough lattice, including the modification from $X_j$ to $(-1)^{N_j}$ along the perpendicular boundary spins, and $B_p$ are all the plaquette terms which fit in the rough lattice. Notice that these modified symmetry operators commute and are properly quantized. By construction, this is a blend of $U(1)_e \times U(1)_m$.

This blend can be equivalently formulated in the modified Villain Hamiltonian formalism \cite{Cheng:2022sgb,Fazza:2022fss,Seifnashri:2026ema,Lu:2026jnq}, which were first introduced as Euclidean lattice models \cite{Gross:1990ub,Sulejmanpasic:2019ytl,Gorantla:2021svj}.  
In this setting, $\phi_j$ is $\mathbb{R}$-valued, and we denote its conjugate momentum by $p_j$ (which unlike $N_j$ is not quantized). They obey $[\phi_j , p_{j'} ] = i\delta_{j,j'}$.
We also introduce a $\mathbb{Z}$-valued gauge field $w_{j+1/2}$ on every link, together with its conjugate momentum $\tilde\phi_{j+1/2}$, satisfying $[w_{j+1/2},\tilde\phi_{j'+1/2}]= - i \delta_{j,j'}$. 
The Villain gauge field satisfies the local constraint $\exp(2\pi i w_{j+1/2})=1$. 
Next, we introduce the Gauss law constraint $\exp(2\pi i p_j +i \tilde\phi_{j-1/2} - i\tilde \phi_{j+1/2})=1$ to make $\phi_j$ effectively $2\pi$-periodic. 
The two $U(1)$ unitary operators are then $U_A(\alpha) = \exp(i \alpha \sum_j w_{j+1/2})$ and $U_V(\beta) =  \exp(i \beta \sum_j p_j)$.

We add these Villain bosons along the rough boundary, and keep the qubits $X_j,Z_j$ on the horizontal links as they are. 
We couple the qubits near the rough boundary to the Villain bosons by modifying the two constraints to:
\[
\exp(i\pi p_j + \frac i2 \tilde\phi_{j-1/2} -\frac i2 \tilde\phi_{j+1/2}) = X_j,\\
\exp(2\pi i w_{j+1/2}) = Z_j Z_{j+\frac 12} Z_{j+1}.
\]
(A fermionic version of this modification were recently discussed in Refs.~\cite{Lu:2026itw,Dharanikota:2026vik}.)
The blend of $U(1)_e\times U(1)_m$ are now:
\[
\tilde{U}_e(\theta) = \exp(i\theta \sum_v \frac{1}{4}(1-A_v) +i{\theta \over 2}\sum_j p_j), \\
\tilde{U}_m(\theta) = \exp(
i\theta \sum_p \frac{1}{4}(1-B_p)+i\theta \sum_jw_{j+1/2}),
\]
satisfying $\tilde U_e(2\pi)=1$ and $\tilde U_m(2\pi)=1$.
\end{proof}
Note that the above convention for a rough boundary condition can be dualized to give a blend in the presence of a smooth boundary.

\end{document}